\documentclass[10pt,reqno]{article}
\usepackage{fullpage} 
\usepackage{amssymb} 
\usepackage{graphicx} 
\usepackage{amsmath} 
\usepackage{amsthm,amssymb,latexsym} 
\usepackage{amstext, amsfonts,amsbsy} 
\usepackage{enumerate} 
\usepackage{upgreek} 
\usepackage{xcolor} 
\usepackage{accents} 
\usepackage{mathtools} 
\usepackage{float} 
\usepackage{tikz}
\usetikzlibrary{matrix,graphs,arrows,positioning,calc,decorations.markings,shapes.symbols,shapes.geometric}
\usepackage[pdftex,bookmarks,colorlinks,breaklinks]{hyperref}  
\definecolor{dullmagenta}{rgb}{0.4,0,0.4}   
\definecolor{darkblue}{rgb}{0,0,0.4}
\hypersetup{linkcolor=red,citecolor=blue,filecolor=dullmagenta,urlcolor=darkblue} 

\newtheorem{theorem}{Theorem}

\newtheorem{lemma}[theorem]{Lemma}

\newtheorem{remark}[theorem]{Remark}
\newcommand{\Pain}[1]{\text{P}_{\mathrm{#1}}}
\newcommand{\dPain}[1]{\text{P}\left(\mathrm{#1}\right)}

\allowdisplaybreaks

\theoremstyle{definition}

\theoremstyle{remark}

\numberwithin{equation}{section}

\begin{document}

{\noindent\Large\bf On the Recurrence Coefficients for the $q$-Laguerre Weight and Discrete Painlev\'e Equations}
\phantom{AAA}\medskip
\\
\textbf{Jie Hu}\\
Department of Mathematics, Jinzhong University, Yuci District, Jinzhong, Shanxi, China\\
E-mail: \href{mailto:hujie_0610@163.com.}{\texttt{hujie\_0610@163.com}}\\[5pt]
\textbf{Anton Dzhamay}\footnote{Corresponding author}\\
Beijing Institute of Mathematical Sciences and Applications (BIMSA), Beijing, China and \\ 
School of Mathematical Sciences, The University of Northern Colorado, Greeley, CO 80526, USA\\
E-mail: \href{mailto:adzham@bimsa.cn}{\texttt{adzham@bimsa.cn}}\\[5pt]
\textbf{Yang Chen}\\
Faculty of Science and Technology, Department of Mathematics, University of Macau,
 Avenida da Universidade, Taipa, Macau, China\\
E-mail: \href{mailto:chenyayang57@gmail.com}{\texttt{chenyayang57@gmail.com}}\\[5pt]

\date{\today}

\begin{abstract}
	We study the dependence of recurrence coefficients in the three-term recurrence relation for orthogonal polynomials
	with a certain deformation of the $q$-Laguerre weight on the degree parameter $n$. We show that this dependence is 
	described by a discrete Painlev\'e equation on the family of $A_{5}^{(1)}$ Sakai surfaces, but this equation is
	different from the standard examples of discrete Painlev\'e equations of this type and instead is a composition 
	of two such. This case study is a good illustration of the effectiveness of a recently proposed geometric 
	identification scheme for discrete Painlev\'e equations.  
\end{abstract}

\begin{flushleft}
\emph{Secondary}:
39A45, 
33D45, 
39A13 
\end{flushleft}

\section{Introduction } 
\label{sec:introduction}

Orthogonal polynomials play an important role in many fields of Mathematics and Mathematical Physics, such as the Random Matrix Theory \cite{Meh:2004:RM}, Approximation Theory, Stochastic Processes, and others. They form a natural basis for expansions of solutions of partial differential or difference equations. There are many connections between orthogonal polynomials and Painlev\'{e} equations, see the recent monograph \cite{Van:2018:OPPE} and the references therein. 
In particular, for semi-classical weights, coefficients of the three-term recurrence relation for discrete and continuous orthogonal polynomials with respect to the semi-classical weight satisfy Painlev\'{e}-type equations, that can be both differential equations w.r.t.~one of the parameters in the weight function, 
\cite{Mag:1995:PDERCSOP}, or discrete w.r.t.~the degree $n$ of the polynomial, see many examples in \cite{Van:2022:OPTLPE}.

One common issue in working with discrete or differential Painlev\'e equations is that of the coordinates --- the same dynamic can take a very 
simple or a very complicated form based on the choice of a coordinate system. This also makes matching the equation with some of the standard examples a very nontrivial task. Arguably the most effective approach to this problem is provided by the algebro-geometric approach to Painlev\'e equations that is primarily due to H.~Sakai, \cite{Sak:2001:RSAWARSGPE}. In this approach the coordinates can be essentially removed from the picture and the equation corresponds to a conjugacy class of a translation element in some affine Weyl group (the symmetry group of the equation). From that point of view, different choices of coordinates correspond to different geometric realizations of the configuration spaces for the same dynamic. The identification question can then be answered on the purely algebraic level, which can be further extended to the construction of an explicit change of coordinates reducing the equation in question to some standard form. 

This procedure was formalized in \cite{DzhFilSto:2020:RCDOPWHWDPE}, and since then a few more examples were considered using this approach, e.g.,  \cite{LiDzhFilZha:2022:RRGLCOPDPEDSS}, \cite{DzhFilSto:2022:DERCSOPTRPEGA}, as well as our previous paper \cite{HuDzhChe:2020:PLUEDPE} where 
we considered the Laguerre unitary ensemble defined by the weight $w(x;\alpha) = x^{\alpha}e^{-x}$, over $[0,t]$ and showed that 
its recurrence coefficients satisfy one of the standard difference Painlev\'e equations on the $D_{5}^{(1)}$-family of Sakai surfaces.  In the present 
paper we consider the $q$-difference case defined by the deformed $q$-Laguerre weight \eqref{eq:q-weight}, supported on $[0,\infty)]$, 
introduced by Chen and Griffin \cite{CheGri:2015:LDEAFDQW}, 
\begin{equation}\label{eq:q-weight}
w(x,\alpha,t;q)=\frac{x^{\alpha}}{((q-1)x; q)_{\infty}\left((q-1)\frac{t}{x};q\right)_{\infty}},\;\;\;t\geq0,\;\alpha>-1,\;0<q<1,
\end{equation}
where $(a;q)_{\infty}$ is the usual $q$-Pochhammer symbol
\begin{align*}
  (a;q)_{\infty}:=\prod_{j=0}^{\infty}(1-a q^{j}).
\end{align*}

The physical motivation to consider such a weight is the following. First, in the limit as $t\rightarrow0^{+}$ this weight 
reduces to the $q$-Laguerre weight \cite{Moa:1981:QLP}, and so the weight \eqref{eq:q-weight}
is a one-parameter deformation of this $q$-Laguerre weight. If we set $\alpha=0$ and  
$t={q}/{(1-q)^2}$, the corresponding orthogonal polynomials are the Stieltjes-Wigert polynomials and the weight 
essentially becomes log-normal, i.e., $w(x) =  \exp(-c(\ln(x))^2)$, where $c>0$, which is related to a physical 
problem of localization. Moreover, this weight generates the indeterminate moment problems, namely, the same set of 
moments are generated by different weights. For $\alpha\neq0$, the corresponding orthogonal polynomials were considered 
by Askey \cite{Ask:1989:OPTF}. For more details, please see \cite{CheGri:2015:LDEAFDQW}, as well as \cite[Section 2]{CheLaw:1998:DZSOP}, 
and references therein.

Special cases of this weight are also known to be related with Painlev\'e and discrete Painlev\'e equations. In particular, 
in the limit as $q\rightarrow1^{-}$ this weight transforms to the weight $x^{\alpha}{e}^{-x}{e}^{-t/x}$, in which case the recurrence coefficients
are related to solutions of the Painlev\'{e} III equation w.r.t.~the $t$-variable, \cite{CheIts:2010:PSLSHRME}, and also satisfy one of the 
standard discrete Painlev\'e equations on the $D_{6}^{(1)}$ surface w.r.t.~discrete degree variable $n$, \cite{LiDzhFilZha:2022:RRGLCOPDPEDSS}.
And for the semiclassical variation of the $q$-Laguerre weight
\begin{equation}\label{eq:semical}
  \frac{x^{\alpha}(-p/{x^2}; q^2)_{\infty}}{(-x^2; q^2)_{\infty}\left(-q^2/x^2; q^2\right)_{\infty}}, x\in[0, \infty), p\in[0, q^{-\alpha}], a\geq0,
\end{equation}
and the $q$-analogue of the Laguerre weight $\frac{x^{\alpha}}{(-x^2; q^2)_{\infty}}$, the recurrence coefficients can be 
expressed in terms of solution of the $q$-Painlev\'{e} V equation \cite{BoeVan:2010:DPERCSLP,BoeVan:2015:VSQPTRC}.
Thus, understanding the type of discrete Painlev\'e equations governing the dynamic of recurrence coefficients for this generalized
weight is a very interesting question. In \cite{CheGri:2015:LDEAFDQW}, the authors suggested that this can be related to $\alpha q$-$\Pain{IV}$
or $\alpha q$-$\Pain{V}$ equations, but this turns out not to be the case. In fact, as we show below, the resulting recurrence
corresponds to a new discrete Painlev\'e equation that is a combination of two standrad dynamics It is very difficult to see such identifications
without using the geometric approach proposed in \cite{DzhFilSto:2020:RCDOPWHWDPE}.

Let us now briefly review the derivation of the recurrence relation that we are interested in, following \cite{CheGri:2015:LDEAFDQW}. 
Consider a family $\{P_{n}(x)\}$ of \emph{monic} 
polynomials of degree $n$ that are orthogonal with respect to the deformed $q$-Laguerre weight function $w(x,\alpha,t;q)$ (\ref{eq:q-weight}) on $[0,\infty)$:
\begin{equation*}
  \int_{0}^{\infty}P_{m}(x)P_{n}(x)w(x,\alpha,t; q)dx=\delta_{m,n}h_{n},
\end{equation*}
where $\delta_{i,j}$ is the Kronecker delta and $h_{n}$ is the square of the $L^{2}$ norm of $P_{n}(x)$ w.r.t.~this weight. It is well-known that
orthogonal polynomials satisfy the three-term recurrence relation
\begin{equation*}
xP_n(x)=P_{n+1}(x)+\alpha_nP_n(x)+\beta_{n}P_{n-1}(x),
\end{equation*} with initial conditions $P_{-1}(x)=0$ and $P_0(x)=1$. It is convenient to parameterize the 
recurrence coefficients $\alpha_{n}$ and $\beta_{n}$ that we are interested in using auxiliary variables $R_n$ and $r_n$ by
\begin{align*}
 q^{2n+\alpha}\alpha_n=&\frac{1-q^n}{1-q}+\frac{1-q^{n+\alpha+1}}{q(1-q)}+q^{n-1}t\big(R_n+(1-q)\sum_{j=0}^{n-1}R_j\big),\\
q^{2n-1}\beta_n=&\frac{1}{q^{2\alpha+2n}}\frac{1-q^n}{1-q}\frac{1-q^{n+\alpha}}{1-q}+\frac{1-q^n}{q^{\alpha+1}} t+\frac{q^n}{q^{\alpha+1}} t  r_n + \frac{1}{q^{2\alpha+n+1}}t\sum_{j=0}^{n-1}R_j,
\end{align*}
where
\begin{align*}
 &R_n=\frac{1}{h_n}\int_{0}^{\infty}P_{n}(y)P_{n}(y/q)\frac{w(y,\alpha,t;q)}{y}{d}y, \\
 &r_n=\frac{1}{h_{n-1}}\int_{0}^{\infty}P_{n}(y)P_{n-1}(y/q)\frac{w(y,\alpha,t;q)}{y}{d}y.
\end{align*}

The relation with discrete Painlev\'e equations is then given by the following theorems. First, there are the recurrence relations obtained by 
Chen and Griffin that describe the evolution, after yet another reparameterization, of the variables $R_{n}$ and $r_{n}$.

\begin{theorem}[{\cite[Theorem 1.4]{CheGri:2015:LDEAFDQW}}]\label{thm:recurrence}
Let 
\begin{align*}
  x_{n}=\frac{q^{n+\alpha}(1-q)}{R_n}, \quad y_{n}=q^{n}(1-r_n),\quad T=\frac{(1-q)^2}{q}t,\quad Q=q^{\alpha}.
\end{align*}
Then the quantities $x_{n}$ and $y_n$ satisfy the following system of difference equations:
\begin{equation}\label{eq:qxn-back-for}
\left\{
  \begin{aligned}
	&(x_{n}y_{n}-1)(x_{n-1}y_{n}-y_{n}) = q^{2n}Q\;T\frac{(y_{n}- 1)(y_{n}-1/T)}{q^{n}-y_{n}},\\
	&(x_{n}y_{n}-1)(x_{n}y_{n+1}-1) = - q^{2n+1}Q\frac{(x_n-1)(x_n-T)}{x_{n}}.
\end{aligned}
\right.
\end{equation}
\end{theorem}

Our main result is the identification of this dynamic with a discrete Painlev\'e equation on the $A_{5}^{(1)}$-family of Sakai surfaces 
(that is different from possible discrete Painlev\'e equations suggested in \cite{CheGri:2015:LDEAFDQW}). 
Recall that  the full symmetry group of the $A_{5}^{(1)}$-family of Sakai surfaces is the extended affine Weyl group of type 
$E_{3}^{(1)}$ (that is sometimes also denoted by $(A_{2}+A_{1})^{(1)}$). This group is described by a \emph{disconnected} Dynkin diagram shown on 
Figure~\ref{fig:dynkin-diags} (on the right).
	\begin{figure}[ht]
	    \begin{equation*}\label{eq:std-symm-roots}			
		    \raisebox{-25.0pt}{\begin{tikzpicture}[elt/.style={circle, draw=black!100, thick, inner sep=0pt, minimum size=2mm}]
	            \begin{scope}[yshift=0.2cm]
					\node[draw=none,minimum size=1.7cm,regular polygon,regular polygon sides=6, rotate=30] (d) {};			
					\foreach \x in {1,2,...,6}
						\path (d.corner \x) node (d\x) [elt] {};
					 \draw [black] (d1) -- (d2) -- (d3) -- (d4) -- (d5) -- (d6) -- (d1);	
					 \node at ($(d1.north) + (0.0,0.2)$) {$\delta_{0}$};
					 \node at ($(d2.west) + (-0.2,0.0)$) {$\delta_{1}$};
					 \node at ($(d3.west) + (-0.2,0.0)$) {$\delta_{2}$};
					 \node at ($(d4.south) + (0.0,-0.2)$) {$\delta_{3}$};
					 \node at ($(d5.east) + (+.2,0.0)$) {$\delta_{4}$};
					 \node at ($(d6.east) + (+0.2,0.0)$) {$\delta_{5}$};
				\end{scope}					
	            \begin{scope}[xshift=5cm]
					\path (-0.6,-0.3) node (a1) [elt] {}
					( 0.6,-0.3) node (a2) [elt] {}
					( 0,0.67) node (a0) [elt] {}
					( 1.5,0.3) node (a3) [elt] {}
					( 2.5,0.3) node (a4) [elt] {};
		            \draw [black] (a1) -- (a2) -- (a0) -- (a1);
					\draw [black,double distance=2pt] (a3) -- (a4);
		            \node at ($(a1.west) + (-0.3,0.0)$) {$\alpha_{1}$};
		            \node at ($(a2.east) + (0.3,0.0)$) {$\alpha_{2}$};
		            \node at ($(a0.north) + (0,0.3)$) {$\alpha_{0}$};
		            \node at ($(a3.north) + (0,0.3)$) {$\alpha_{3}$};
		            \node at ($(a4.north) + (0,0.3)$) {$\alpha_{4}$};
				\end{scope}	
		    \end{tikzpicture}} 
	    \end{equation*}
		\caption{Affine Dynkin diagrams of type $A_{5}^{(1)}$ (left) and $E_{3}^{(1)}$ (right).} 
		\label{fig:dynkin-diags}
	\end{figure}
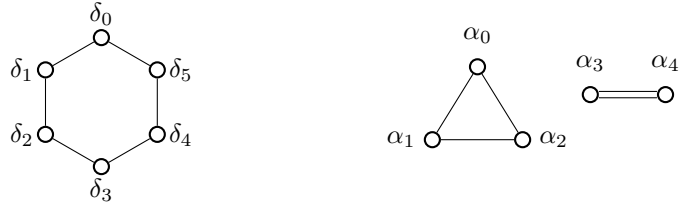

There are two standard examples of discrete Painlev\'e equations on this surface family that naturally correspond to translation elements
on one of the connected components of the Dynkin diagram $E_{3}^{(1)}$. The first translation acts on the symmetry 
root lattice as
\begin{equation}\label{eq:dP-A5-a-trans}
\psi_*:\upalpha=\langle\alpha_0,\alpha_1,\alpha_2;\alpha_3,\alpha_{4}\rangle\mapsto
\psi_*(\upalpha)=\upalpha+ \langle 0,0,0;-1,1 \rangle\delta,
\end{equation}
and so is a translation on the $A_{1}^{(1)}$-sublattice,  and the second one as
\begin{equation}\label{eq:dP-A5-b-trans}
\phi_*:\upalpha=\langle\alpha_0,\alpha_1,\alpha_2;\alpha_3,\alpha_{4}\rangle\mapsto
\phi_*(\upalpha)=\upalpha+ \langle 0,1,-1;0,0 \rangle\delta,
\end{equation}
and so is a translation on the $A_{2}^{(1)}$-sublattice.
Based on the translation vectors (actually, their equivalence classes), we use the notation $[000\overline{1}1]$ and $[01\overline{1}00]$ for these 
two translation dynamics. The dynamic $[000\overline{1}1]$ is also known as a $q$-$\Pain{IV}$ equation since it  has a continuous limit to the 
standard differential $\Pain{IV}$ equation, and for the similar reason the dynamic $[01\overline{1}00]$ is known as the $q$-$\Pain{III}$ equation,
see \cite{Sak:2001:RSAWARSGPE}. The existence of such continuous limits can be also seen from the degeneration cascade in the Sakai classification
scheme shown on Figure~\ref{fig:Sakai-clsc-surf}.

\begin{figure}[ht]{\small
	\begin{tikzpicture}[>=stealth,scale=0.935]
		\node (e8e) at (2,4) {$\left(A_{0}^{(1)}\right)^{\text{e}}$};
		\node (a1qa) at (16,4) {$\left(A_{7}^{(1)}\right)^{\text{q}}$};
		\node (e8q) at (2,2) {$\left(A_{0}^{(1)*}\right)^{\text{q}}$};
		\node (e7q) at (4,2) {$\left(A_{1}^{(1)}\right)^{\text{q}}$};
		\node (e6q) at (6,2) {$\left(A_{2}^{(1)}\right)^{\text{q}}$};
		\node (d5q) at (8,2) {$\left(A_{3}^{(1)}\right)^{\text{q}}$};
		\node (a4q) at (10,2) {$\left(A_{4}^{(1)}\right)^{\text{q}}$};
		\node (a21q) at (12,2) {$\left(A_{5}^{(1)}\right)^{\text{q}}$};
		\node (a11q) at (14,2) {$\left(A_{6}^{(1)}\right)^{\text{q}}$};
		\node (a1q) at (16,2) {$\left(A_{7}^{(1)}\right)^{\text{q}}$};
		\node (a0q) at (18,2) {$\left(A_{8}^{(1)}\right)^{\text{q}}$};
		\node (e8d) at (4,0) {$\left(A_{0}^{(1)**}\right)^{\text{d}}$};
		\node (e7d) at (6,0) {$\left(A_{1}^{(1)*}\right)^{\text{d}}$};
		\node (e6d) at (8,0) {$\left(A_{2}^{(1)*}\right)^{\text{d}}$};
		\node (d4d) at (10,0) {$\left(D_{4}^{(1)}\right)^{\text{d,c}}$};
		\node (a3d) at (12,0) {$\left(D_{5}^{(1)}\right)^{\text{d,c}}$};
		\node (a11d) at (14,0) {$\left(D_{6}^{(1)}\right)^{\text{d,c}}$};
		\node (a1d) at (16,0) {$\left(D_{7}^{(1)}\right)^{\text{d,c}}$};
		\node (a0d) at (18,0) {$\left(D_{8}^{(1)}\right)^{\text{d,c}}$};
		\node (a2d) at (14,-2) {$\left(E_{6}^{(1)}\right)^{\text{d,c}}$};
		\node (a1da) at (16,-2) {$\left(E_{7}^{(1)}\right)^{\text{d,c}}$};
		\node (a0da) at (18,-2) {$\left(E_{8}^{(1)}\right)^{\text{c}}$};
		\draw[->] (e8e) -> (e8q);	\draw[->] (a1qa) -> (a0d);
		\draw[->] (e8q) -> (e7q); 	\draw[->] (e8q) -> (e8d);
		\draw[->] (e7q) -> (e6q); 	\draw[->] (e7q) -> (e7d);
		\draw[->] (e6q) -> (d5q); 	\draw[->] (e6q) -> (e6d);
		\draw[->] (d5q) -> (a4q); 	\draw[->] (d5q) -> (d4d);
		\draw[->] (a4q) -> (a21q); 	\draw[->] (a4q) -> (a3d);
		\draw[->] (a21q) -> (a11q); \draw[->] (a21q) -> (a11d); \draw[->] (a21q) -> (a2d);
		\draw[->] (a11q) -> (a1q); 	\draw[->] (a11q) -> (a1d); 	\draw[->] (a11q) -> (a1qa); 	\draw[->] (a11q) -> (a1da);
		\draw[->] (a1q) -> (a0q); 	\draw[->] (a1q) -> (a0d);	\draw[->] (a1q) -> (a0da);
		\draw[->] (e8d) -> (e7d);
		\draw[->] (e7d) -> (e6d);
		\draw[->] (e6d) -> (d4d);
		\draw[->] (d4d) -> (a3d);
		\draw[->] (a3d) -> (a11d);	\draw[->] (a3d) -> (a2d);
		\draw[->] (a11d) -> (a1d);	\draw[->] (a11d) -> (a1da);
		\draw[->] (a1d) -> (a0d);	\draw[->] (a1d) -> (a0da);
		\draw[->] (a2d) -> (a1da);	\draw[->] (a1da) -> (a0da);
		\node [purple] at ($(a21q.north) + (0,0.2)$) {q-$P_{\text{IV}}$, q-$P_{\text{III}}$};
		\node [blue] at ($(d4d.south) + (0,-0.1)$) {$P_{\text{VI}}$};
		\node [blue] at ($(a3d.south) + (-0.3,-0.1)$) {$P_{\text{V}}$};
		\node [blue] at ($(a11d.south) + (0,-0.1)$) {$P_{\text{III}}$};
		\node [blue] at ($(a1d.south) + (0,-0.1)$) {$P_{\text{III}}$};
		\node [blue] at ($(a0d.south) + (-0.2,-0.1)$) {$P_{\text{III}}$};
		\node [blue] at ($(a2d.south) + (0,-0.1)$) {$P_{\text{IV}}$};
		\node [blue] at ($(a1da.south) + (0,-0.1)$) {$P_{\text{II}}$};
		\node [blue] at ($(a0da.south) + (0,-0.1)$) {$P_{\text{I}}$};
	\end{tikzpicture}}
	\caption{Surface-type classification scheme for Painlev\'e equations}
\label{fig:Sakai-clsc-surf}
\end{figure}
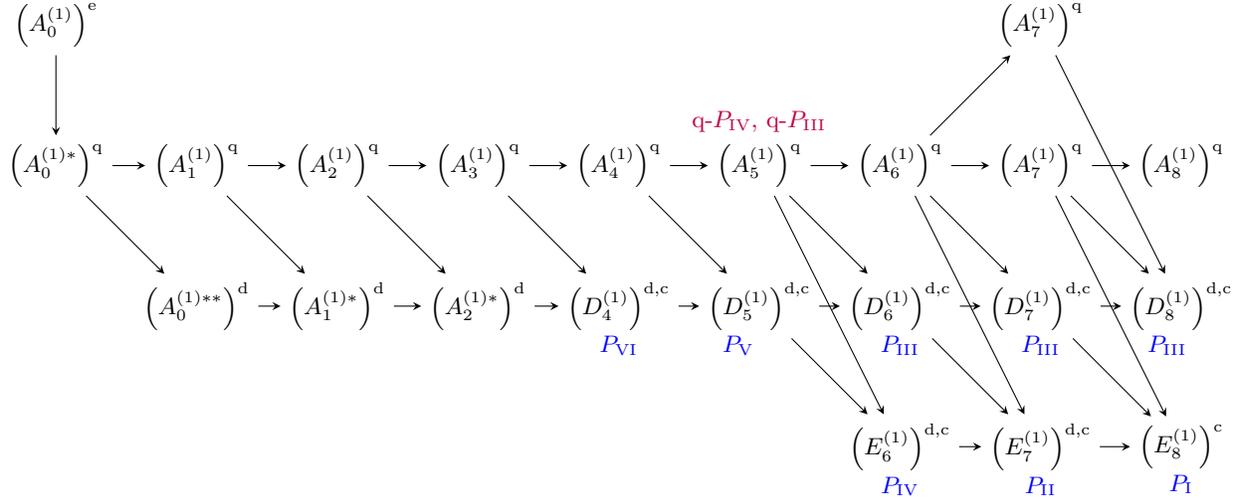

In \cite{KajNouYam:2017:GAPE} the authors gave a careful description of the standard realizations 
of Sakai surfaces for each of the cases shown on Figure~\ref{fig:Sakai-clsc-surf}, and it became
the standard reference on this subject. For the $A_{5}^{(1)}$-family of surfaces there are 
two such realizations: the (a)-model \cite[8.2.7]{KajNouYam:2017:GAPE} is better suited for the 
dynamic \eqref{eq:dP-A5-a-trans} and the (b)-model \cite[8.2.10]{KajNouYam:2017:GAPE} for the
dynamic \eqref{eq:dP-A5-b-trans}. Of course, both families are birationally equivalent. 
To make the paper self-contained, we 
briefly describe them in Section~\ref{sec:dP-A5-surf-std}.

In contrast to many other examples of discrete Painlev\'e equations that occur in the study of 
orthogonal polynomials, the recurrence \eqref{eq:qxn-back-for} is not conjugated
to any of these two standard equations and instead is essentially their composition with 
translations in both components. Namely, our main result is the following Theorem.

\begin{theorem}\label{thm:main} The discrete dynamic given by \eqref{eq:qxn-back-for} defines a discrete Painlev\'e equation 
	 $[01\overline{1}\overline{1}1]$	on the 
	$A_{5}^{(1)}$ family of Sakai surfaces that is equivalent to the composition $\psi\circ\phi = \phi\circ\psi$ of two standard mappings. 
\end{theorem}

We prove this Theorem, as well as give the explicit form of the mapping \eqref{eq:qLag-std} on the standard (b)-model realization of the $A_{5}^{(1)}$-family 
and the explicit birational change of coordinates \eqref{eq:Charlier2Sakai-coord} transforming our dynamic to this form, in Section~\ref{sec:identification}.

It is easy to see that the mappings $\phi$, $\psi$, and $\varphi$ are non-conjugate. Indeed, for that it is enough to look at the 
Jordan block structure of the induced linear maps on the Picard lattice. For equation ${[000\overline{1}1]}$
	it is $J(-1,1)^{\oplus3}\oplus J(1,1)^{\oplus4}\oplus J(1,3)$, for equation ${[01\overline{1}00]}$ it is 
	$J(1,1)^{\oplus3}\oplus J(1,3)\oplus J(e^{2 \pi \mathfrak{i}/3},1)^{\oplus 2}\oplus J(e^{4 \pi \mathfrak{i}/3},1)^{\oplus 2}$,
	and for equation $[01\overline{1}\overline{1}1]$ it is 
	$J(-1,1)\oplus J(1,1)^{\oplus2}\oplus J(1,3)\oplus J(e^{ \pi \mathfrak{i}/3},1)
	\oplus J(e^{2 \pi \mathfrak{i}/3},1)\oplus J(e^{4 \pi \mathfrak{i}/3},1)\oplus J(e^{5 \pi \mathfrak{i}/3},1)$.

\begin{remark} Given that the equation ${[000\overline{1}1]}$ has a continuous limit to the 
standard differential $\Pain{IV}$ equation and the equation $[01\overline{1}00]$ has a continuous limit to the 
standard differential $\Pain{III}$ equation, it would be very interesting to see if the equation 
$[01\overline{1}\overline{1}1]$ has a good continuous limit. We plan to consider this question in the follow-up work.	
\end{remark}

The paper is organized as follows. In Section~\ref{sec:dP-A5-surf-std} we collect some basic data about the 
q-$\dPain{E_{3}^{(1)}/A_{5}^{(1)}}$ surface families, and in Section~\ref{sec:identification} we perform the detailed
study of recurrence \eqref{eq:qxn-back-for}, along the lines of \cite{DzhFilSto:2020:RCDOPWHWDPE}, 
and establish our main result. The final section is a brief conclusion. 

\section{Discrete Painlev\'e Equations on the q-$\dPain{E_{3}^{(1)}/A_{5}^{(1)}}$ Surfaces} 
\label{sec:dP-A5-surf-std}
To make the paper self-contained, we begin by briefly describing some basic geometric data for discrete Painlev\'e equations for the $A_{5}^{(1)}$ surface 
family, following \cite{KajNouYam:2017:GAPE}. We also give birational representations of the extended affine Weyl symmetry group 
$\tilde{W}\left(E_{3}^{(1)}\right)$, which is essentially the same as in \cite{KajNouYam:2017:GAPE}, but some maps differ by a choice of normalizations.
There are two natural realizations of the $A_{5}^{(1)}$-family via 
the configuration of the blow-up points: the (a)-model \cite[8.2.7]{KajNouYam:2017:GAPE} and the (b)-model \cite[8.2.10]{KajNouYam:2017:GAPE}.
We begin with the (a)-model.

\subsection{The q-$\dPain{E_{3}^{(1)}/A_{5}^{(1)};a}$ Surface Family} 
\label{sub:case-a}

This family corresponds to the choice of root bases for the surface and symmetry sub-lattices shown on 
Figure~\ref{fig:roots-bases-a}.
\begin{figure}[ht]
\begin{equation}\label{eq:d-roots-KNY-a}
	\raisebox{-45pt}{\begin{tikzpicture}[
			elt/.style={circle,draw=black!100,thick, inner sep=0pt,minimum size=2mm}]
			\path (-0.6,-0.3) node (a1) [elt] {}
			( 0.6,-0.3) node (a2) [elt] {}
			( 0,0.67) node (a0) [elt] {}
			( -0.6,-1.2) node (a3) [elt] {}
			( 0.6,-1.2) node (a4) [elt] {};
            \draw [black] (a1) -- (a2) -- (a0) -- (a1);
			\draw [black,double distance=2pt] (a3) -- (a4);
            \node at ($(a1.west) + (-0.3,0.0)$) {$\alpha_{1}$};
            \node at ($(a2.east) + (0.3,0.0)$) {$\alpha_{2}$};
            \node at ($(a0.north) + (0,0.3)$) {$\alpha_{0}$};
            \node at ($(a3.south) + (0,-0.3)$) {$\alpha_{3}$};
            \node at ($(a4.south) + (0,-0.3)$) {$\alpha_{4}$};			
	\end{tikzpicture}}\quad 
			\begin{aligned}
			\alpha_{0} &= \mathcal{H}_{f} + \mathcal{H}_{g} - \mathcal{E}_{2367},	\\		
			\alpha_{1} &= \mathcal{H}_{f} + \mathcal{H}_{g} - \mathcal{E}_{1468}, \\
			\alpha_{2} &= \mathcal{E}_{6} - \mathcal{E}_{5},  \\
			\alpha_{3} &=\mathcal{H}_{f} + 2 \mathcal{H}_{g} - \mathcal{E}_{123568}\\
			\alpha_{4} &= \mathcal{H}_{f}-\mathcal{E}_{47},\\
			\end{aligned}
	\qquad	\qquad 	
	\raisebox{-45pt}{\begin{tikzpicture}[
			elt/.style={circle,draw=black!100,thick, inner sep=0pt,minimum size=2mm}]
		\path 				(0,1)	node 	(d0) 	[elt, label=above:{$\delta_{0}$} ] {}
		        (-{sqrt(3)/2},1/2)	node 	(d1) 	[elt, label=left:{$\delta_{1}$} ] {}
		        (-{sqrt(3)/2},-1/2) node  	(d2)	[elt, label=left:{$\delta_{2}$} ] {}
		     	   			(0,-1)	node  	(d3) 	[elt, label=below:{$\delta_{3}$} ] {}
		        ({sqrt(3)/2},-1/2) 	node  	(d4) 	[elt, label=right:{$\delta_{4}$} ] {}
		        ({sqrt(3)/2},1/2) 	node 	(d5) 	[elt, label=right:{$\delta_{5}$} ] {};
		\draw [black,line width=1pt ] (d0) -- (d1) -- (d2) -- (d3) -- (d4) -- (d5)--(d0);
	\end{tikzpicture}}\quad 
			\begin{aligned}
			\delta_{0} &= \mathcal{H}_{f} - \mathcal{E}_{12},	\\		
			\delta_{1} &= \mathcal{E}_{2} - \mathcal{E}_{3}, \\
			\delta_{2} &= \mathcal{H}_{g} - \mathcal{E}_{24},  \\
			\delta_{3} &=\mathcal{H}_{f}- \mathcal{E}_{56},\\
			\delta_{4} &= \mathcal{H}_{g}-\mathcal{E}_{17},\\
			\delta_{5} &= \mathcal{E}_{1} - \mathcal{E}_{8};
			\end{aligned}
\end{equation}
	\caption{The symmetry (left) and the surface (right) root bases for the q-$\dPain{E_{3}^{(1)}/A_{5}^{(1)};a}$}
	\label{fig:roots-bases-a}
\end{figure}

\subsubsection{The point configuration} 
\label{ssub:points-a}

The decomposition of the anti-canonical divisor class into the classes of irreducible components $\delta_{i}$ above,
\begin{equation*}
	-\mathcal{K}_{\mathcal{X}} = [H_{f}-E_{1}-E_{2}] + [E_{2} - E_{3}] + [H_{g} - E_{2} - E_{4}] + [H_{f} - E_{5} - E_{6}] + [H_{g} - E_{1} - E_{7}] + 
	[E_{1} - E_{8}],
\end{equation*}
can be realized by the point configuration on Figure~\ref{fig:KNY-a-pt-conf}. In \cite{KajNouYam:2017:GAPE} all 
$q$-type configurations are parameterized using the same set of $10$ parameters $\kappa_{1},\kappa_{2},\nu_{1},\ldots,\nu_{8}$.
Using these parameters, in case (a) the points are
\begin{equation*}\label{eq:base-pt-KNY-a}
	\begin{aligned}
		&p_{1}\left(\frac{1}{f}=0,g=0\right)\leftarrow p_{8}\left(\frac{1}{f}=0,f g = -\frac{\kappa_{1}}{\nu_{1} \nu_{8}}\right),\ 
		p_{2}\left(\frac{1}{f}=0,\frac{1}{g}=0\right)\leftarrow p_{3}\left(\frac{1}{f}=0,\frac{f}{g} = -\nu_{2} \nu_{3}\right),\\
		&p_{4}\left(f = \nu_{4},\frac{1}{g}=0\right),\ p_{5}\left(f=0,g=\frac{\nu_{5}}{\kappa_{2}}\right),\ 
		p_{6}\left(f=0,g=\frac{\nu_{6}}{\kappa_{2}}\right),\ p_{7}\left(f=\frac{\kappa_{1}}{\nu_{7}},g=0\right).
	\end{aligned}
\end{equation*}
However, it is easy to see that in this case the true number of parameters is only \emph{four}: using the action of the gauge group of
M\"obius transformations we can make some of the point coordinates vanish, and there is still a residual two-parameter rescaling 
acting on the remaining six non-zero coordinates. As usual, it is convenient to consider the canonical parameters known as 
the \emph{root variables}.

\begin{figure}[ht]
	\begin{center}		
	\begin{tikzpicture}[>=stealth,basept/.style={circle, draw=red!100, fill=red!100, thick, inner sep=0pt,minimum size=1.2mm}]
	\begin{scope}[xshift=0cm,yshift=0cm]
	\draw [black, line width = 1pt] (-0.2,0) -- (3.2,0)	node [pos=0,left] {\small $H_{g}$} node [pos=1,right] {\small $p=0$};
	\draw [black, line width = 1pt] (-0.2,3) -- (3.2,3) node [pos=0,left] {\small $H_{g}$} node [pos=1,right] {\small $p=\infty$};
	\draw [black, line width = 1pt] (0,-0.2) -- (0,3.2) node [pos=0,below] {\small $H_{f}$} node [pos=1,xshift = -7pt, yshift=5pt] {\small $q=0$};
	\draw [black, line width = 1pt] (3,-0.2) -- (3,3.2) node [pos=0,below] {\small $H_{f}$} node [pos=1,xshift = 7pt, yshift=5pt] {\small $q=\infty$};
	\node (p1) at (3,0) [basept,label={[xshift = -7pt, yshift=-3pt] \small $p_{1}$}] {};
	\node (p8) at (3.5,0.5) [basept,label={[xshift = 10pt, yshift=-8pt] \small $p_{8}$}] {};
	\node (p2) at (3,3) [basept,label={[xshift = -7pt, yshift=-15pt] \small $p_{2}$}] {};
	\node (p3) at (3.5,2.5) [basept,label={[xshift = 10pt, yshift=-8pt] \small $p_{3}$}] {};
	\node (p4) at (2,3) [basept,label={[xshift = 0pt, yshift=-15pt] \small $p_{4}$}] {};
	\node (p5) at (0,2) [basept,label={[xshift = 10pt,yshift=-8pt] \small $p_{5}$}] {};
	\node (p6) at (0,0.8) [basept,label={[xshift = 10pt,yshift=-8pt] \small $p_{6}$}] {};
	\node (p7) at (1,0) [basept,label={[xshift = 0pt,yshift=-3pt] \small $p_{7}$}] {};
	\draw [red, line width = 0.8pt, ->] (p8) -- (p1);
	\draw [red, line width = 0.8pt, ->] (p3) -- (p2);
	\end{scope}
	\draw [->] (7,1.5)--(5,1.5) node[pos=0.5, below] {$\operatorname{Bl}_{p_{1}\cdots p_{8}}$};
	\begin{scope}[xshift=9.5cm,yshift=0cm]
	\draw[blue, line width = 1pt] (-1.2,3)--(2.7,3) node [pos=0,left] {\small $H_{g}- E_{2} - E_{4}$};
	\draw[blue, line width = 1pt] (-1.2,0) -- (2.7,0) node [pos=0,left] {\small $H_{g}-E_{1} - E_{7}$};	
	\draw[blue, line width = 1pt] (3,0.3)--(3,2.7) node [pos=0.5,right] {\small $ H_{f} - E_{1} - E_{2}$};
	\draw[blue, line width = 1pt] (-1,-0.2)--(-1,3.2) node [pos=1,above] {\small $H_{f} - E_{5}- E_{6}$};
	\draw[blue, line width = 1pt] (2.2,3.2)--(3.2,2.2) node [pos = 0, above] {\small $ E_{2}- E_{3}$};
	\draw[red,  line width = 1pt] (2.3,2.3)--(3,3) node [pos=0,below left] {\small $E_3$};
	\draw[blue, line width = 1pt] (2.2,-0.2)--(3.2,0.8); \node[blue] at (2.2,-0.5) {\small $ E_{1}- E_{8}$};
	\draw[red,  line width = 1pt] (3,0) -- (2.3,0.7) node [pos=1,above left] {\small $E_{8}$};
	\draw[red,  line width = 1pt] (0.7,2.6) -- (1.4,3.3) node [pos=0,below] {\small $E_{4}$};	
	\draw[red,  line width = 1pt] (-1.3,1.7)--(-0.6,2.4) node [pos=1,right] {\small $E_{5}$};
	\draw[red,  line width = 1pt] (-1.3,0.5)--(-0.6,1.2) node [pos=1,right] {\small $E_{6}$};
	\draw[red,  line width = 1pt] (-0.3,-0.3) -- (0.4,0.4) node [pos=1,above] {\small $E_{7}$};	
	\end{scope}
	\end{tikzpicture}
	\end{center}
	\caption{The model type (a) Sakai surface for the q-$\dPain{E_{3}^{(1)}/A_{5}^{(1)}}$ family} 
	\label{fig:KNY-a-pt-conf}
\end{figure}	

Recall that \emph{root variables} $a_{i}$ are computed using the \emph{period map} 
$\chi: \operatorname{Span}_{\mathbb{Z}}\{\alpha_{i}\}\to \mathbb{C}$, see  \cite{Sak:2001:RSAWARSGPE} 
or \cite{DzhTak:2018:SASGTDPE} for details. In this case the period map is defined using the symplectic form 
$\omega = \frac{df \wedge dg}{fg}$ and the root variables
 $a_{i} = \exp(\chi(\alpha_{i}))$ for this configuration can be expressed in terms of the above parameters as
\begin{equation}\label{eq:root-vars-KNY-a}
	a_{0} = \frac{\kappa_{1} \kappa_{2}}{\nu_{2}\nu_{3}\nu_{6}\nu_{7}},\ a_{1} = \frac{\kappa_{1}\kappa_{2}}{\nu_{1}\nu_{4}\nu_{6}\nu_{8}},\ 
	a_{2} = \frac{\nu_{6}}{\nu_{5}},\ a_{3} = \frac{\kappa_{1} \kappa_{2}^{2}}{\nu_{1}\nu_{2}\nu_{3}\nu_{5}\nu_{6}\nu_{8}},\ 
	a_{4} = \frac{\kappa_{1}}{\nu_{4}\nu_{7}}. 
\end{equation}
These root variables satisfy the constraint
\begin{equation}
	a_{0}a_{1}a_{2} = a_{3}a_{4} = \exp(\chi(\delta = - \mathcal{K}_{\mathcal{X}})) = 
	q = \frac{\kappa_{1}^{2} \kappa_{2}^{2}}{\nu_{1}\nu_{2}\nu_{3}\nu_{4}\nu_{5}\nu_{6}\nu_{7}\nu_{8}},
\end{equation}
where $q$ is the parameter of the dynamic. We can then express the coordinates of the base points using the root variables,
\begin{equation*}\label{eq:base-pt-root-a}
	\begin{aligned}
		&p_{1}\left(\frac{1}{f}=0,g=0\right)\leftarrow p_{8}\left(\frac{1}{f}=0,f g = -\frac{a_{1}a_{4}}{a_{0}}\frac{\nu_{4}^{2}}{\nu_{2}\nu_{3}}\right),\ 
		p_{2}\left(\frac{1}{f}=0,\frac{1}{g}=0\right)\leftarrow p_{3}\left(\frac{1}{f}=0,\frac{f}{g} = -\nu_{2} \nu_{3}\right),\\
		&p_{4}\left(f = \nu_{4},\frac{1}{g}=0\right),\ p_{5}\left(f=0,g=\frac{a_{4}}{a_{0}a_{2}}\frac{\nu_{4}}{\nu_{2}\nu_{3}}\right),\ 
		p_{6}\left(f=0,g=\frac{a_{4}}{a_{0}}\frac{\nu_{4}}{\nu_{2}\nu_{3}}\right),\ p_{7}\left(f=a_{4}\nu_{4},g=0\right),
	\end{aligned}
\end{equation*}
where the free parameters $\nu_{2}\nu_{3}$ and $\nu_{4}$ can be set to $1$ using the rescaling action on the coordinate axes. 


\subsubsection{The extended affine Weyl symmetry group} 
\label{ssub:birrep-a}

Recall that, given a Dynkin diagram, the corresponding Weyl group is defined in terms of generators $w_{i}$ corresponding to the 
nodes $\alpha_{i}$ of the diagram, and with edges of the diagram encoding the relations between these generators. For the affine $E_{3}^{(1)}$, we have
\begin{equation*}
	W\left(E_{3}^{(1)}\right) = W\left(\raisebox{-35pt}{\begin{tikzpicture}[
			elt/.style={circle,draw=black!100,thick, inner sep=0pt,minimum size=1.3ex},scale=0.8]
			\path (-0.6,-0.3) node (a1) [elt] {}
			( 0.6,-0.3) node (a2) [elt] {}
			( 0,0.67) node (a0) [elt] {}
			( -0.6,-1.2) node (a3) [elt] {}
			( 0.6,-1.2) node (a4) [elt] {};
            \draw [black] (a1) -- (a2) -- (a0) -- (a1);
			\draw [black,double distance=2pt] (a3) -- (a4);
            \node at ($(a1.west) + (-0.3,0.0)$) {$\alpha_{1}$};
            \node at ($(a2.east) + (0.3,0.0)$) {$\alpha_{2}$};
            \node at ($(a0.north) + (0,0.3)$) {$\alpha_{0}$};
            \node at ($(a3.south) + (0,-0.3)$) {$\alpha_{3}$};
            \node at ($(a4.south) + (0,-0.3)$) {$\alpha_{4}$};			
	\end{tikzpicture}} \right)
	=
	\left\langle w_{0},\dots, w_{4}\ \left|\ 
	\begin{alignedat}{2}
    w_{i}^{2} = e,\quad  w_{i}\circ w_{j} &= w_{j}\circ w_{i}& &\text{ when 
   				\raisebox{-10pt}{\begin{tikzpicture}[
   							elt/.style={circle,draw=black!100,thick, inner sep=0pt,minimum size=1.5mm}]
   						\path   ( 0,0) 	node  	(ai) [elt] {}
   						        ( 0.5,0) 	node  	(aj) [elt] {};
   						\draw [black] (ai)  (aj);
   							\node at ($(ai.south) + (0,-0.2)$) 	{$\alpha_{i}$};
   							\node at ($(aj.south) + (0,-0.2)$)  {$\alpha_{j}$};
   							\end{tikzpicture}}}\\
w_{i}\circ w_{j} \circ w_{i}&= w_{j}\circ w_{i} \circ w_{j}& &\text{ when 
   				\raisebox{-10pt}{\begin{tikzpicture}[
   							elt/.style={circle,draw=black!100,thick, inner sep=0pt,minimum size=1.5mm}]
   						\path   ( 0,0) 	node  	(ai) [elt] {}
   						        ( 0.5,0) 	node  	(aj) [elt] {};
   						\draw [black] (ai)--(aj);
   							\node at ($(ai.south) + (0,-0.2)$) 	{$\alpha_{i}$};
   							\node at ($(aj.south) + (0,-0.2)$)  {$\alpha_{j}$};
   							\end{tikzpicture}}}							
	\end{alignedat}\right.\right\rangle. 
\end{equation*} 
In our setting this group is represented via actions on $\operatorname{Pic}(\mathcal{X})$ given by reflections in the 
roots $\alpha_{i}$, 
\begin{equation}\label{eq:root-refl}
	w_{i}(\mathcal{C}) = w_{\alpha_{i}}(\mathcal{C}) = \mathcal{C} - 2 
	\frac{\mathcal{C}\bullet \alpha_{i}}{\alpha_{i}\bullet \alpha_{i}}\alpha_{i}
	= \mathcal{C} + \left(\mathcal{C}\bullet \alpha_{i}\right) \alpha_{i},\qquad \mathcal{C}\in \operatorname{Pic(\mathcal{X})}.
\end{equation}
Next, we need to extend this group by the group of the automorphisms of the Dynkin diagram (that corresponds to some re-labeling of the 
symmetry/surface roots) $\operatorname{Aut}(E_3^{(1)})\simeq \operatorname{Aut}(A_{5}^{(1)})\simeq \mathbb{D}_{6}$, where
$\mathbb{D}_{6}$ is the usual dihedral group of the symmetries of a regular hexagon (note that this group acts
on both the symmetry and the surface roots and we describe its action using the standard permutation cycle notation). 
This group is generated by the two 
reflections $\sigma_{0}=(\alpha_{0}\alpha_{2})(\alpha_{3}\alpha_{4})=(\delta_{0}\delta_{1})(\delta_{2}\delta_{5})(\delta_{3}\delta_{4})$
and $\sigma_{1}=(\alpha_{1}\alpha_{2})=(\delta_{0}\delta_{2})(\delta_{3}\delta_{5})$, and it is convenient to also 
consider the rotation $\sigma_{2} = \sigma_{1}\circ \sigma_{0} = (\alpha_{0}\alpha_{1}\alpha_{2})(\alpha_{3}\alpha_{4}) 
= \delta_{0} \delta_{1} \delta_{2} \delta_{3} \delta_{4} \delta_{5}$, see Figure~\ref{fig:automs}.

The action of generators of $\operatorname{Aut}(E_{3}^{(1)})$ on $\operatorname{Pic}(\mathcal{X})$ can be 
realized as compositions of reflections in some other roots in the lattice, e.g.,
\begin{align*}
	\sigma_{0} &= w_{\mathcal{E}_{5} - \mathcal{E}_{7}} \circ 	w_{\mathcal{E}_{4} - \mathcal{E}_{8}} \circ 	
		w_{\mathcal{H}_{g} - \mathcal{E}_{1} - \mathcal{E}_{2}} \circ 
		w_{\mathcal{H}_{g} - \mathcal{E}_{3} - \mathcal{E}_{6}} \circ 
		w_{\mathcal{H}_{f} - \mathcal{E}_{2} - \mathcal{E}_{6}},	\\
	\sigma_{1} &= w_{\mathcal{E}_{5} - \mathcal{E}_{8}} \circ 	
		w_{\mathcal{H}_{g} - \mathcal{E}_{4} - \mathcal{E}_{6}} \circ 
		w_{\mathcal{H}_{f} - \mathcal{E}_{1} - \mathcal{E}_{6}}.	
\end{align*}
The resulting group $\widetilde{W}(E_{3}^{(1)})=\operatorname{Aut}(E_{3}^{(1)})\ltimes W(E_{3}^{(1)})$ is called an  
extended affine Weyl symmetry group and its action on $\operatorname{Pic}(\mathcal{X})$ can 
be further extended to an action on point configurations by elementary birational maps (which lifts to 
isomorphisms $w_{i}: \mathcal{X}_{\mathbf{b}}\to \mathcal{X}_{\overline{\mathbf{b}}}$ on the family of Sakai's surfaces);
this is known as a \emph{birational representation}. We describe it in the following Lemma.

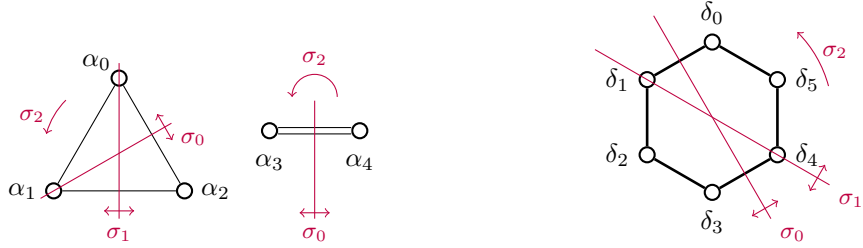
\begin{figure}[ht]
\begin{equation*}
	\raisebox{-45pt}{\begin{tikzpicture}[
			elt/.style={circle,draw=black!100,thick, inner sep=0pt,minimum size=2mm}]
			\path (-0.87,-0.5) node (a1) [elt] {}
			( 0.87,-0.5) node (a2) [elt] {}
			( 0,1) node (a0) [elt] {}
			( 2,0.3) node (a3) [elt] {}
			( 3.2,0.3) node (a4) [elt] {};
            \draw [black] (a1) -- (a2) -- (a0) -- (a1);
			\draw [black,double distance=2pt] (a3) -- (a4);
            \node at ($(a1.west) + (-0.3,0.0)$) {$\alpha_{1}$};
            \node at ($(a2.east) + (0.3,0.0)$) {$\alpha_{2}$};
            \node at ($(a0.north) + (-0.3,0.1)$) {$\alpha_{0}$};
            \node at ($(a3.south) + (0,-0.3)$) {$\alpha_{3}$};
            \node at ($(a4.south) + (0,-0.3)$) {$\alpha_{4}$};		
			\draw [purple, postaction={decorate,decoration={markings, mark=at position 0.95 with {\draw [<->,purple,thin] (0,-2mm) -- (0,2mm);}}}]
			  ($1.2*(a1)$) --  ($0.8*(0.87,0.5)$) node [pos=1,below right] {\small $\sigma_{0}$};
			\draw [purple, postaction={decorate,decoration={markings, mark=at position 0.95 with {\draw [<->,purple,thin] (0,-2mm) -- (0,2mm);}}}]
			  (0,1.2) --  (0,-0.87) node [pos=1,below] {\small $\sigma_{1}$};		
  		  \draw[-latex,purple,->](135:1) arc(135:165:1) node[left,pos=0.5] {\small $\sigma_{2}$};			  		
			\draw [purple, postaction={decorate,decoration={markings, mark=at position 0.95 with {\draw [<->,purple,thin] (0,-2mm) -- (0,2mm);}}}]
			  (2.6,0.7) --  (2.6,-0.87) node [pos=1,below] {\small $\sigma_{0}$};		
  		  \draw[-latex,purple,->] (2.9,0.75) arc(0:180:0.3) node[above,pos=0.5] {\small $\sigma_{2}$};
	\end{tikzpicture}}
	\qquad	\qquad 	\qquad \qquad
	\raisebox{-45pt}{\begin{tikzpicture}[
			elt/.style={circle,draw=black!100,thick, inner sep=0pt,minimum size=2mm}]
		\path 				(0,1)	node 	(d0) 	[elt, label=above:{$\delta_{0}$} ] {}
		        (-{sqrt(3)/2},1/2)	node 	(d1) 	[elt, label=left:{$\delta_{1}$} ] {}
		        (-{sqrt(3)/2},-1/2) node  	(d2)	[elt, label=left:{$\delta_{2}$} ] {}
		     	   			(0,-1)	node  	(d3) 	[elt, label=below:{$\delta_{3}$} ] {}
		        ({sqrt(3)/2},-1/2) 	node  	(d4) 	[elt, label=right:{$\delta_{4}$} ] {}
		        ({sqrt(3)/2},1/2) 	node 	(d5) 	[elt, label=right:{$\delta_{5}$} ] {};
		\draw [black,line width=1pt ] (d0) -- (d1) -- (d2) -- (d3) -- (d4) -- (d5)--(d0);
		\draw [purple, postaction={decorate,decoration={markings, mark=at position 0.95 with {\draw [<->,purple,thin] (0,-2mm) -- (0,2mm);}}}]
		  ($1.8*(d1)$) --  ($1.8*(d4)$) node [pos=1,below right] {\small $\sigma_{1}$};
		\draw [purple, postaction={decorate,decoration={markings, mark=at position 0.95 with {\draw [<->,purple,thin] (0,-2mm) -- (0,2mm);}}}]
		  ($1.8*(-{sqrt(3)/4},3/4)$) --  ($1.8*({sqrt(3)/4},-3/4)$) node [pos=1,below right] {\small $\sigma_{0}$};
		  \draw[-latex,purple,->](15:1.6) arc(15:45:1.6) node[right,pos=0.6] {\small $\sigma_{2}$};
	\end{tikzpicture}}
\end{equation*}
	\caption{The automorphisms of Dynkin diagrams $E_{3}^{(1)}$ and $A_{5}^{(1)}$}
	\label{fig:automs}
\end{figure}


\begin{lemma}\label{lem:bir-repE3-a}
	Generators of the extended affine Weyl group $\widetilde{W}(E_{3}^{(1)})$ transform an initial point configuration
	in the (a)-model,
	expressed in both the KNY parameters and the root variables as
	\begin{equation}\label{eq:base-pt-a}
    \left(\begin{matrix}
    \kappa_{1} & \nu_{1} & \nu_{3} & \nu_{5} & \nu_{7}\\
    \kappa_{2} & \nu_{2} & \nu_{4} & \nu_{6} & \nu_{8}
    \end{matrix}\ ;\ \begin{matrix} f\\ g \end{matrix}\right) \sim
	\left(\begin{matrix}
    a_{0} & a_{1} & a_{2} \\
    a_{3} & a_{4}
    \end{matrix}\ ;\ \begin{matrix} f\\ g \end{matrix}\right), 
	\end{equation}	
	by changing the parameters and the affine coordinates $(f,g)$ as follows:
	\begin{align*}
	w_0&:
		\left(\begin{matrix}
		\kappa_{1} & \nu_{1}\left(\frac{\nu_{2}\nu_{3}\nu_{6}\nu_{7}}{\kappa_{1}\kappa_{2}}\right)^{3} & 
			\nu_{3} & \nu_{5} & \nu_{7}\\
		\frac{\nu_{2}\nu_{3}\nu_{6}\nu_{7}}{\kappa_{1}} & \nu_{2} & \nu_{4} & \frac{\kappa_{1}\kappa_{2}}{\nu_{2}\nu_{3}\nu_{7}} & \nu_{8}
		\end{matrix}\ ;\ \begin{matrix} 
		\frac{\kappa_{1}\kappa_{2}f(\nu_{7}(f + \nu_{2}\nu_{3}g) - \kappa_{1})}{\nu_{2}\nu_{3}\nu_{7}(\nu_{6}(\nu_{7}f - \kappa_{1})+\kappa_{1}\kappa_{2}g)} \\[5pt]	
		\frac{\kappa_{1}^{2}\kappa_{2}g(\kappa_{2}(f + \nu_{2}\nu_{3}g)- \nu_{2}\nu_{3}\nu_{6})}{(\nu_{2}\nu_{3})^{2}\nu_{6}\nu_{7}(\nu_{6}(\nu_{7}f - 		\kappa_{1})+\kappa_{1}\kappa_{2}g)}\end{matrix} \right) \\
	&\qquad \sim 	
		\left(\begin{matrix}
		\frac{1}{a_{0}} & a_{0}a_{1} & a_{0}a_{2} \\
		a_{3} & a_{4}
		\end{matrix}\ ;\ \begin{matrix} \frac{a_{0} f(f + g - a_{4})}{f + a_{0}g - a_{4}}\\[5pt] 
		\frac{a_{0}g(a_{0}(f + g) - a_{4})}{(f + a_{0}g - a_{4})} \end{matrix}\right), \\[7pt]
	w_{1}&:
	    \left(\begin{matrix}
	    \kappa_{1} & \nu_{1}\left(\frac{\kappa_{1}\kappa_{2}}{\nu_{1}\nu_{4}\nu_{6}\nu_{8}}\right)^{3} & 
		\nu_{3} & \nu_{5} & \nu_{7}\\
	    \kappa_{2} \left(\frac{\kappa_{1}\kappa_{2}}{\nu_{1}\nu_{4}\nu_{6}\nu_{8}}\right)^{2} & \nu_{2} & \nu_{4} & 	
			\frac{\kappa_{1}\kappa_{2}}{\nu_{1}\nu_{4}\nu_{6}} & \nu_{8}
	    \end{matrix}\ ;\ \begin{matrix} \frac{\nu_{4}\nu_{6}f(\nu_{1}\nu_{8}(f - \nu_{4})g+ \kappa_{1})}{\kappa_{1}(\kappa_{2}(f - \nu_{4})g+ \nu_{4}\nu_{6})}\\[5pt] 
		\frac{(\nu_{1}\nu_{6}\nu_{8})^{2}\nu_{4}g(\kappa_{2}(f - \nu_{4})g+ \nu_{4}\nu_{6})}{\kappa_{1}\kappa_{2}^{2}((\nu_{1}\nu_{6}\nu_{8}f - \kappa_{1}\kappa_{2})g+ \kappa_{1}\nu_{6})}
		 \end{matrix}\right) \\
	&\qquad \sim 			
		\left(\begin{matrix}
	    a_{0}a_{1} & \frac{1}{a_{1}} & a_{1}a_{2} \\
	    a_{3} & a_{4}
	    \end{matrix}\ ;\ \begin{matrix} \frac{f \left(a_0 (f-1) g+a_1 a_4\right)}{a_1 \left(a_0 (f-1) g+a_4\right)}\\[5pt] 
		\frac{g \left(a_0 (f-1) g+a_4\right)}{a_1 \left(a_0 g \left(f-a_1\right)+a_1 a_4\right)}\end{matrix}\right), 	\\[7pt]
	w_{2}&:
	    \left(\begin{matrix}
	    \kappa_{1} & \nu_{1} & \nu_{3} & \nu_{5} & \nu_{7}\\
	    \frac{\kappa_{2} \nu_{5}}{\nu_{6}} & \nu_{2} & \nu_{4} &\frac{\nu_{5}^{2}}{\nu_{6}} & \nu_{8}
	    \end{matrix}\ ;\ \begin{matrix} f\\[5pt] g \end{matrix}\right) \sim
		\left(\begin{matrix}
	    a_{0}a_{2} & a_{1}a_{2} & \frac{1}{a_{2}} \\
	    a_{3} & a_{4}
	    \end{matrix}\ ;\ \begin{matrix} f\\[5pt] g \end{matrix}\right), 	\\[7pt]	
	w_{3}&:
	    \left(\begin{matrix}
	    \frac{\kappa_{1}^{3}\kappa_{2}^{4}}{(\nu_{1}\nu_{2}\nu_{3}\nu_{5}\nu_{6}\nu_{8})^{2}} & 
		\nu_{1}  & \nu_{3} & 	\nu_{5}  &  \nu_{7}\\
	    \frac{(\nu_{1}\nu_{2}\nu_{3}\nu_{5}\nu_{6}\nu_{8})^{2}}{\kappa_{1}^{2}\kappa_{2}^{3}} & \nu_{2} & \nu_{4} & \nu_{6}   & \nu_{8}
	    \end{matrix}\ ; \right. \\
		&\qquad \left.\ \begin{matrix} 
		\frac{\kappa _1 \kappa _2^2 f \left(g \left(\nu _1 \nu _5 \nu _8 \left(\kappa _2 \left(f+g \nu _2 \nu _3\right)-\nu _2 \nu _3 \nu _5\right)
		-\kappa _1 \kappa _2^2\right)+\kappa _1 \kappa _2 \nu _5\right) \left(g \left(\nu _1 \nu _6 \nu _8 \left(\kappa _2 \left(f+g \nu _2 \nu _3\right)
		-\nu _2 \nu _3 \nu _6\right)-\kappa _1 \kappa _2^2\right)+\kappa _1 \kappa _2 \nu _6\right)}
		{\left(\nu _1 \nu _8\right)^{2}\left(\nu _2 \nu _3 \nu _5 \nu _6\right) 
		\left(f g \kappa _2^2+\nu _2 \nu _3 \left(g \kappa _2-\nu _5\right) \left(g \kappa _2-\nu _6\right)\right) \left(f g \nu _1 \nu _5 \nu _6 \nu _8+\kappa _1 
		\left(g \kappa _2-\nu _5\right) \left(g \kappa _2-\nu _6\right)\right)}
		\\[5pt] 
		\frac{g \left(\kappa _1 \kappa _2\right){}^2 \left(f g \kappa _2^2+\nu _2 \nu _3 \left(g \kappa _2-\nu _5\right) \left(g \kappa _2-\nu _6\right)\right)}{\left(\nu _2 
		\nu _3\right){}^2 \left(\nu _1 \nu _5 \nu _6 \nu _8\right) \left(f g \nu _1 \nu _5 \nu _6 \nu _8+\kappa _1 \left(g \kappa _2-\nu _5\right) 
		\left(g \kappa _2-\nu _6\right)\right)}
		\end{matrix}\right) \\ 
		&\qquad \sim 
		\left(\begin{matrix}
	    a_{0} & a_{1} & a_{2} \\
	    \frac{1}{a_{3}} & a_{3}^{2}a_{4}
	    \end{matrix}\ ;\ \begin{matrix} 
		\frac{a_0^2 a_1 a_2 f \left(a_0 g \left(-a_1+f+g\right)-a_4 \left(g-a_1\right)\right) 
		\left(a_0 a_2 g \left(-a_1 a_2+f+g\right)-a_4 \left(g-a_1 a_2\right)\right)}{\left(a_0 g \left(a_0 a_2 (f+g)-\left(a_2+1\right) a_4\right)+a_4^2\right) 
		\left(a_0 g \left(a_4 \left(f-a_1 \left(a_2+1\right)\right)+a_0 a_1 a_2 g\right)+a_1 a_4^2\right)}\\[5pt] 
		\frac{a_0 a_1^2 a_2 g \left(a_0 g \left(a_0 a_2 (f+g)-\left(a_2+1\right) a_4\right)+a_4^2\right)}{
			a_4 \left(a_0 g \left(a_4 \left(f-a_1 \left(a_2+1\right)\right)+a_0 a_1 a_2 g\right)+a_1 a_4^2\right)}  \end{matrix}\right), 	\\[7pt]	
	w_{4}&:
	    \left(\begin{matrix}
	    \frac{(\nu_{4}\nu_{7})^{2}}{\kappa_{1}} & \nu_{1}  & \nu_{3} & \nu_{5}  & \nu_{7} \\
	    \frac{\kappa_{1}^{2}\kappa_{2}}{(\nu_{4}\nu_{7})^{2}} & \nu_{2} & \nu_{4} & \nu_{6} & \nu_{8}
	    \end{matrix}\ ;\ \begin{matrix} \frac{\nu_{4} \nu_{7} f}{\kappa_{1}}\\[5pt] \frac{\nu_{4}\nu_{7}^{2}(f - \nu_{4})g}{\kappa_{1}(\nu_{7} f - \kappa_{1})} \end{matrix}\right) \sim
		\left(\begin{matrix}
	    a_{0} & a_{1} & a_{2} \\
	    a_{3}a_{4}^{2} &\frac{1}{a_{4}} 
	    \end{matrix}\ ;\ \begin{matrix} \frac{f}{a_{4}}\\[5pt] \frac{(f-1)g}{a_{4}(f-a_{4})} \end{matrix}\right),  \\[7pt]
	\sigma_{0}&:
	    \left(\begin{matrix}
	   \frac{\nu_{1}\nu_{2}\nu_{3}\nu_{4}\nu_{5}\nu_{6}\nu_{7}\nu_{8}}{\kappa_{1}\kappa_{2}^{2}} & 
	   		\frac{\kappa_{1}\kappa_{2}\nu_{2}\nu_{3}\nu_{5}\nu_{7}}{\nu_{1}(\nu_{4}\nu_{6}\nu_{8})^{2}} & \nu_{3} & \nu_{5} & \nu_{7}\\
	    \frac{\kappa_{2}\nu_{2}\nu_{3}\nu_{5}\nu_{7}}{\nu_{1}\nu_{4}\nu_{6}\nu_{8}} & \nu_{2} & \nu_{4} & \frac{\nu_{2}\nu_{3}\nu_{5}\nu_{6}\nu_{7}}{\kappa_{1}\kappa_{2}} & \nu_{8}
	    \end{matrix}\ ;\ \begin{matrix} \frac{g \nu _1 \nu _4 \nu _6 \nu _8 \left(\kappa _2 \left(f+g \nu _2 \nu _3\right)-
			\nu _2 \nu _3 \nu _6\right)}{\kappa _1 \kappa _2 \left(g \kappa _2-\nu _6\right)}\\[5pt] 
			-\frac{f \nu _1 \nu _4 \nu _6^2 \nu _8}{\kappa _1 \kappa _2 \nu _2 \nu _3 \left(g \kappa _2-\nu _6\right)} \end{matrix}\right)\\
		&\qquad \sim
		\left(\begin{matrix}
	    \frac{1}{a_{2}} & \frac{1}{a_{1}} & \frac{1}{a_{0}} \\
	    \frac{1}{a_{4}} & \frac{1}{a_{3}}
	    \end{matrix}\ ;\ \begin{matrix} \frac{g \left(a_0 (f+g)-a_4\right)}{a_1 \left(a_0 g-a_4\right)}\\[5pt] 
		-\frac{a_4 f}{a_1 \left(a_0 g-a_4\right)} \end{matrix}\right), 	\\[7pt]	
	\sigma_{1}&:
		    \left(\begin{matrix}
		    \frac{(\nu_{4}\nu_{7})^{2}}{\kappa_{1}} & \frac{(\nu_{2}\nu_{3}\nu_{6}\nu_{7})^{2}}{\kappa_{1}\kappa_{2}\nu_{4}\nu_{5}\nu_{8}} &
			\nu_{3} & \nu_{5} & \nu_{7}\\
			   \frac{\nu_{1}\nu_{5}\nu_{8}(\nu_{2}\nu_{3}\nu_{6})^{2}}{\kappa_{1}\kappa_{2}^{2}\nu_{4}} & \nu_{2} & \nu_{4} &
			   \frac{\nu_{1}\nu_{4}\nu_{5}\nu_{6}\nu_{8}}{\kappa_{1}\kappa_{2}} & \nu_{8}
		    \end{matrix}\ ;\ \begin{matrix} \frac{\nu _4 \left(g \kappa _2 \left(f-\nu _4\right)+\nu _4 \nu _6\right)}{f \nu _6}\\[5pt]
			\frac{f g \kappa _2^2 \nu _4}{\left(\nu _2 \nu _3\right){}^2 \nu _6 \left(g \kappa _2-\nu _6\right)} \end{matrix}\right)\\
		& \qquad \sim
		\left(\begin{matrix}
	    \frac{1}{a_{0}} & \frac{1}{a_{2}} & \frac{1}{a_{1}} \\
	    \frac{1}{a_{3}} & \frac{1}{a_{4}}
	    \end{matrix}\ ;\ \begin{matrix} \frac{a_0 (f-1) g+a_4}{a_4 f}\\[5pt] \frac{a_0^2 f g}{a_4 \left(a_0 g-a_4\right)}\end{matrix}\right), 	\\[7pt]				
	\sigma_{2}&:
	    \left(\begin{matrix}
	    \frac{\kappa_{1}\kappa_{2}^{2}\nu_{4}\nu_{7}}{\nu_{1}\nu_{2}\nu_{3}\nu_{6}\nu_{8}} & \frac{(\nu_{2}\nu_{3}\nu_{6}\nu_{7})^{2}}{\kappa_{1}\kappa_{2}\nu_{4}\nu_{5}\nu_{8}} 
			& \nu_{3} & \nu_{5} & \nu_{7}\\
	    \frac{(\nu_{2}\nu_{3})^{2}\nu_{5}\nu_{6}}{\kappa_{2}\nu_{4}^{2}} & \nu_{2} & \nu_{4} & \frac{\kappa_{1}\kappa_{2}\nu_{5}}{\nu_{1}\nu_{4}\nu_{6}\nu_{8}} & \nu_{8}
	    \end{matrix}\ ;\ \begin{matrix} -\frac{\kappa _2 \nu _4 \left(g \left(f \nu _1 \nu _6 \nu _8-\kappa _1 \kappa _2\right)+\kappa _1 
		\nu _6\right)}{g \nu _1 \nu _2 \nu _3 \nu _6 \nu _8 \left(g \kappa _2-\nu _6\right)}\\[5pt] 
			\frac{f g \kappa _2^2 \nu _4}{\left(\nu _2 \nu _3\right){}^2 \nu _6 \left(g \kappa _2-\nu _6\right)}\end{matrix}\right)\\
		&\qquad  \sim
		\left(\begin{matrix}
	    a_{2} & a_{0} & a_{1} \\
	    a_{4} & a_{3}
	    \end{matrix}\ ;\ \begin{matrix} -\frac{a_0 \left(f-a_1\right) g +a_1 a_4}{g \left(a_0 g-a_4\right)}\\[5pt] 
		\frac{a_0^2 f g}{a_4 \left(a_0 g-a_4\right)}\end{matrix}\right). 		\end{align*}
\end{lemma}

The proof is standard, see \cite{DzhFilSto:2020:RCDOPWHWDPE,DzhTak:2018:SASGTDPE} for similar computations explained in detail. Note that reflection automorphisms $\sigma_{0}$ and $\sigma_{1}$ change the sign of the symplectic form, and so $a_{i}$ changes to $a_{i}^{-1}$. 
We also remark that to compute the evolution of the parameters $\kappa_{i}$ and $\gamma_{j}$, we use 
the evolution of the root variables $a_{i}$. Recall that there
is redundancy in the full set of parameters $\kappa_{i}$ and $\gamma_{j}$ in \cite{KajNouYam:2017:GAPE} for this point configuration. We first express
some of them in terms of the others and the root variables $a_{i}$ using \eqref{eq:root-vars-KNY-a} --- 
there are various choices but they are equivalent up to changing the overall 
normalization. We put
\begin{equation}\label{eq:KNY2root-a}
	\kappa_{1} = a_{4}\cdot \nu_{4}\nu_{7},\quad \kappa_{2} = \frac{a_{0}a_{2}}{a_{4}}\cdot \frac{\nu_{2}\nu_{3}\nu_{5}}{\nu_{4}},\quad 
	\nu_{1} = \frac{a_{0}}{a_{1}}\cdot \frac{\nu_{2}\nu_{3}\nu_{7}}{\nu_{4}\nu_{8}},\quad \nu_{6} = a_{2}\cdot \nu_{5}.
\end{equation}
Then, for $w_{0}$, the root variables evolve as $\overline{a}_{0} = 1/a_{0}$, $\overline{a}_{1} = a_{0} a_{1}$, $\overline{a}_{2} = a_{0} a_{2}$,
$\overline{a}_{3} = a_{3}$, $\overline{a}_{4} = a_{4}$, and so, for example, 
\begin{equation*}
	\overline{\kappa}_{2} =  \frac{\overline{a}_{0}\overline{a}_{2}}{\overline{a}_{4}} \cdot \frac{\nu_{2}\nu_{3}\nu_{5}}{\nu_{4}}
	= \frac{a_{2}}{a_{4}}\cdot \frac{\nu_{2}\nu_{3}\nu_{5}}{\nu_{4}} = 
	\frac{\kappa_{2}}{a_{0}}  = \frac{\nu_{2}\nu_{3}\nu_{6}\nu_{7}}{\kappa_{1}},
\end{equation*}
where we used \eqref{eq:root-vars-KNY-a} again at the last step. Other computations are similar.


\subsubsection{The discrete Painlev\'e equation $[000\overline{1}1]$ on the q-$\dPain{E_{3}^{(1)}/A_{5}^{(1)};a}$ Surface Family} 
\label{ssub:dP-a}

In \cite{KajNouYam:2017:GAPE} the evolution for all multiplicative cases is induced by the evolution of parameters
\begin{equation}\label{eq:KNY-par-evol}
	\overline{\kappa}_{1} = \frac{\kappa_{1}}{q},\quad \overline{\kappa}_{2} = q \kappa_{2},\qquad\text{where }
	q= \frac{\kappa_{1}^{2}\kappa_{2}^{2}}{\nu_{1}\nu_{2}\nu_{3}\nu_{4}\nu_{5}\nu_{6}\nu_{7}\nu_{8}}.
\end{equation}
This induces the following evolution on the root variables \eqref{eq:root-vars-KNY-a}:
\begin{equation*}
	\overline{a}_{0} = a_{0},\quad \overline{a}_{1} = a_{1}, \quad \overline{a}_{2} = a_{2},\quad \overline{a}_{3} = q a_{3},\quad \overline{a}_{4} = \frac{a_{4}}{q},
\end{equation*}
that in turn induces the translation \eqref{eq:dP-A5-a-trans} on the symmetry roots,
\begin{equation*}
\psi_*:\upalpha=\langle\alpha_0,\alpha_1,\alpha_2;\alpha_3,\alpha_{4}\rangle\mapsto
\psi_*(\upalpha)=\upalpha+ \langle 0,0,0;-1,1 \rangle\delta.
\end{equation*}
Using the standard techniques, again see \cite{DzhFilSto:2020:RCDOPWHWDPE,DzhTak:2018:SASGTDPE}, this translation element can be expressed in terms of the generators
of $\widetilde{W}(E_{3}^{(1)})$ as $\psi = \sigma_{2}^{3}\circ w_{3} = w_{4} \circ \sigma_{2}^{3}$, resulting in the mapping 
\begin{equation*}
	\psi(f,g) = \left( - \frac{\nu_{2} \nu_{3} \nu_{4}(\kappa_{2} g - \nu_{5})(\kappa_{2}g - \nu_{6})}{\kappa_{2}^{2} f g},
	\frac{\nu_{2}\nu_{3}\nu_{4}\nu_{5}\nu_{6}\nu_{7}(\nu_{1}\nu_{5}\nu_{6}\nu_{8}fg + \kappa_{1}(\kappa_{2}g - \nu_{5})(\kappa_{2}g - \nu_{6}))}{
	(\kappa_{1}\kappa_{2})^{2}g(\kappa_{2}^{2}(f + \nu_{2}\nu_{3}g)g - \nu_{2}\nu_{3}(\kappa_{2}(\nu_{5}+\nu_{6})g - \nu_{5}\nu_{6}))}\right),
\end{equation*}
which is equivalent to \cite[(8.11)]{KajNouYam:2017:GAPE},
\begin{equation*}
	f \overline{f} = - \nu_{2} \nu_{3} \nu_{4}\frac{\left(g - \frac{\nu_{5}}{\kappa_{2}}\right)\left(g - \frac{\nu_{6}}{\kappa_{2}}\right)}{g},\quad 
	g \underline{g} = \frac{\kappa_{1}}{\nu_{1}\nu_{2}\nu_{3}\nu_{8}}\, \frac{f - \frac{\kappa_{1}}{\nu_{7}}}{f - \nu_{4}},
\end{equation*}
see also \cite[(2.13--2.14)]{Sak:2007:PDPETLF} and the $q$-$\Pain{IV}\to\Pain{IV}$ degeneration in 	\cite{Sak:2001:RSAWARSGPE}.



\subsection{The q-$\dPain{E_{3}^{(1)}/A_{5}^{(1)};b}$ Surface Family} 
\label{sub:case-b}

Let us now turn our attention to the (b)-model that corresponds to the choice of root bases for the surface and symmetry sub-lattices shown on 
Figure~\ref{fig:roots-bases-b}.
\begin{figure}[ht]
\begin{equation}\label{eq:d-roots-KNY-b}
	\raisebox{-45pt}{\begin{tikzpicture}[
			elt/.style={circle,draw=black!100,thick, inner sep=0pt,minimum size=2mm}]
			\path (-0.6,-0.3) node (a1) [elt] {}
			( 0.6,-0.3) node (a2) [elt] {}
			( 0,0.67) node (a0) [elt] {}
			( -0.6,-1.2) node (a3) [elt] {}
			( 0.6,-1.2) node (a4) [elt] {};
            \draw [black] (a1) -- (a2) -- (a0) -- (a1);
			\draw [black,double distance=2pt] (a3) -- (a4);
            \node at ($(a1.west) + (-0.3,0.0)$) {$\alpha_{1}$};
            \node at ($(a2.east) + (0.3,0.0)$) {$\alpha_{2}$};
            \node at ($(a0.north) + (0,0.3)$) {$\alpha_{0}$};
            \node at ($(a3.south) + (0,-0.3)$) {$\alpha_{3}$};
            \node at ($(a4.south) + (0,-0.3)$) {$\alpha_{4}$};			
	\end{tikzpicture}}\quad 
			\begin{aligned}
			\alpha_{0} &= \mathcal{H}_{f} + \mathcal{H}_{g} - \mathcal{E}_{2367},	\\		
			\alpha_{1} &= \mathcal{H}_{f}  - \mathcal{E}_{48}, \\
			\alpha_{2} &= \mathcal{H}_{g} - \mathcal{E}_{15},  \\
			\alpha_{3} &=\mathcal{H}_{f} + \mathcal{H}_{g} - \mathcal{E}_{2358}\\
			\alpha_{4} &= \mathcal{H}_{f} + \mathcal{H}_{g} - \mathcal{E}_{1467},\\
			\end{aligned}
	\qquad	\qquad 	
	\raisebox{-45pt}{\begin{tikzpicture}[
			elt/.style={circle,draw=black!100,thick, inner sep=0pt,minimum size=2mm}]
		\path 				(0,1)	node 	(d0) 	[elt, label=above:{$\delta_{0}$} ] {}
		        (-{sqrt(3)/2},1/2)	node 	(d1) 	[elt, label=left:{$\delta_{1}$} ] {}
		        (-{sqrt(3)/2},-1/2) node  	(d2)	[elt, label=left:{$\delta_{2}$} ] {}
		     	   			(0,-1)	node  	(d3) 	[elt, label=below:{$\delta_{3}$} ] {}
		        ({sqrt(3)/2},-1/2) 	node  	(d4) 	[elt, label=right:{$\delta_{4}$} ] {}
		        ({sqrt(3)/2},1/2) 	node 	(d5) 	[elt, label=right:{$\delta_{5}$} ] {};
		\draw [black,line width=1pt ] (d0) -- (d1) -- (d2) -- (d3) -- (d4) -- (d5)--(d0);
	\end{tikzpicture}}\quad 
			\begin{aligned}
			\delta_{0} &= \mathcal{H}_{f} - \mathcal{E}_{12},	\\		
			\delta_{1} &= \mathcal{E}_{2} - \mathcal{E}_{3}, \\
			\delta_{2} &= \mathcal{H}_{g} - \mathcal{E}_{24},  \\
			\delta_{3} &=\mathcal{H}_{f}- \mathcal{E}_{56},\\
			\delta_{4} &= \mathcal{E}_{6}-\mathcal{E}_{7},\\
			\delta_{5} &= \mathcal{H}_{g} - \mathcal{E}_{68};
			\end{aligned}
\end{equation}
	\caption{The symmetry (left) and the surface (right) root bases for the q-$\dPain{E_{3}^{(1)}/A_{5}^{(1)};b}$}
	\label{fig:roots-bases-b}
\end{figure}

\subsubsection{The point configuration} 
\label{ssub:points-b}

The decomposition of the anti-canonical divisor class into the classes of irreducible components $\delta_{i}$ above,
\begin{equation*}
	-\mathcal{K}_{\mathcal{X}} = [H_{f}-E_{1}-E_{2}] + [E_{2} - E_{3}] + [H_{g} - E_{2} - E_{4}] + [H_{f} - E_{5} - E_{6}] +  [E_{6} - E_{7}] + [H_{g} - E_{6} - E_{8}],
\end{equation*}
can be realized by the point configuration on Figure~\ref{fig:KNY-b-pt-conf} where the coordinates of the base points are again given in terms of the 
the \cite{KajNouYam:2017:GAPE}  parameters $\kappa_{i}$, $\nu_{j}$:
\begin{equation*}\label{eq:base-pt-KNY-b}
	\begin{aligned}
		&p_{1}\left(\frac{1}{f}=0,g=\frac{1}{\nu_{1}}\right),\ 
		p_{2}\left(\frac{1}{f}=0,\frac{1}{g}=0\right)\leftarrow p_{3}\left(\frac{1}{f}=0,\frac{f}{g} = -\nu_{2} \nu_{3}\right),p_{4}\left(f = \nu_{4},\frac{1}{g}=0\right),\\ 
		&p_{5}\left(f=0,g=\frac{\nu_{5}}{\kappa_{2}}\right),\ 
		p_{6}\left(f=0,g=0\right)\leftarrow p_{7}\left(f=0,\frac{g}{f} = -\frac{\nu_{6} \nu_{7}}{\kappa_{1}\kappa_{2}}\right),\ 
		p_{8}\left(f=\frac{\kappa_{1}}{\nu_{8}},g=0\right).
	\end{aligned}
\end{equation*}
Using the same symplectic form 
$\omega = \frac{df \wedge dg}{fg}$ as in the (a)-case, we see that in this case the root variables are slightly different,
\begin{equation}\label{eq:root-vars-KNY-b}
	a_{0} = \frac{\kappa_{1} \kappa_{2}}{\nu_{2}\nu_{3}\nu_{6}\nu_{7}},\ a_{1} = \frac{\kappa_{1}}{\nu_{4}\nu_{8}},\ 
	a_{2} = \frac{\kappa_{2}}{\nu_{1}\nu_{5}},\ 
	a_{3} = \frac{\kappa_{1} \kappa_{2}}{\nu_{2}\nu_{3}\nu_{5}\nu_{8}},\ 
	a_{4} = \frac{\kappa_{1} \kappa_{2}}{\nu_{1}\nu_{4}\nu_{6}\nu_{7}},\ 
\end{equation}
but they satisfy the same constraint $a_{0}a_{1}a_{2} = a_{3}a_{4} = q$ as before.
Using the root variables, the coordinates of the base points are
\begin{equation*}\label{eq:base-pt-root-b}
	\begin{aligned}
		&p_{1}\left(\frac{1}{f}=0,g=\frac{1}{\nu_{1}}\right),\ 
		p_{2}\left(\frac{1}{f}=0,\frac{1}{g}=0\right)\leftarrow p_{3}\left(\frac{1}{f}=0,\frac{f}{g} = -\frac{a_{4}}{a_{0}}\nu_{1}\nu_{4}\right),
		p_{4}\left(f = \nu_{4},\frac{1}{g}=0\right),\\ 
		&p_{5}\left(f=0,g=\frac{1}{a_{2}\nu_{1}}\right),\ 
		p_{6}\left(f=0,g=0\right)\leftarrow p_{7}\left(f=0,\frac{g}{f} = -\frac{1}{a_{4}\nu_{1}\nu_{4}}\right),\ 
		p_{8}\left(f=a_{1}\nu_{4},g=0\right).
	\end{aligned}
\end{equation*}
where the free parameters $\nu_{1}$ and $\nu_{4}$ can again be set to $1$ using the rescaling action. 

\begin{figure}[ht]
	\begin{center}		
	\begin{tikzpicture}[>=stealth,basept/.style={circle, draw=red!100, fill=red!100, thick, inner sep=0pt,minimum size=1.2mm}]
	\begin{scope}[xshift=0cm,yshift=0cm]
	\draw [black, line width = 1pt] (-0.2,0) -- (3.2,0)	node [pos=0,left] {\small $H_{g}$} node [pos=1,right] {\small $p=0$};
	\draw [black, line width = 1pt] (-0.2,3) -- (3.2,3) node [pos=0,left] {\small $H_{g}$} node [pos=1,right] {\small $p=\infty$};
	\draw [black, line width = 1pt] (0,-0.2) -- (0,3.2) node [pos=0,below] {\small $H_{f}$} node [pos=1,xshift = -7pt, yshift=5pt] {\small $q=0$};
	\draw [black, line width = 1pt] (3,-0.2) -- (3,3.2) node [pos=0,below] {\small $H_{f}$} node [pos=1,xshift = 7pt, yshift=5pt] {\small $q=\infty$};
	\node (p1) at (3,1) [basept,label={[xshift = -7pt, yshift=-3pt] \small $p_{1}$}] {};
	\node (p2) at (3,3) [basept,label={[xshift = -7pt, yshift=-15pt] \small $p_{2}$}] {};
	\node (p3) at (3.5,2.5) [basept,label={[xshift = 10pt, yshift=-8pt] \small $p_{3}$}] {};
	\node (p4) at (1,3) [basept,label={[xshift = 0pt, yshift=-15pt] \small $p_{4}$}] {};
	\node (p5) at (0,2) [basept,label={[xshift = 10pt,yshift=-8pt] \small $p_{5}$}] {};
	\node (p6) at (0,0) [basept,label={[xshift = 10pt,yshift=-2pt] \small $p_{6}$}] {};
	\node (p7) at (-0.5,0.5) [basept,label={[xshift = -7pt,yshift=-2pt] \small $p_{7}$}] {};
	\node (p8) at (2,0) [basept,label={[xshift = 0pt,yshift=-3pt] \small $p_{8}$}] {};
	\draw [red, line width = 0.8pt, ->] (p3) -- (p2);
	\draw [red, line width = 0.8pt, ->] (p7) -- (p6);
	\end{scope}
	\draw [->] (7,1.5)--(5,1.5) node[pos=0.5, below] {$\operatorname{Bl}_{p_{1}\cdots p_{8}}$};
	\begin{scope}[xshift=9.5cm,yshift=0cm]
	\draw[blue, line width = 1pt] (-1.2,3)--(2.7,3) node [pos=0,left] {\small $H_{g}- E_{2} - E_{4}$};
	\draw[blue, line width = 1pt] (-0.6,0) -- (3.3,0) node [pos=1,right] {\small $H_{g}-E_{6} - E_{8}$};
	\draw[blue, line width = 1pt] (3,-0.3)--(3,2.7) node [pos=0,below] {\small $ H_{f} - E_{1} - E_{2}$};
	\draw[blue, line width = 1pt] (-1,0.3)--(-1,3.2) node [pos=1,above] {\small $H_{f} - E_{5}- E_{6}$};
	\draw[blue, line width = 1pt] (2.2,3.2)--(3.2,2.2) node [pos = 1, right] {\small $ E_{2}- E_{3}$};
	\draw[red,  line width = 1pt] (2.3,2.3)--(3,3) node [pos=0,below left] {\small $E_3$};
	\draw[blue, line width = 1pt] (-1.2,0.8)--(-0.2,-0.2) node [pos=0,left] {\small $E_6-E_7$};
	\draw[red,  line width = 1pt] (-1,0) -- (-0.3,0.7) node [pos=1,above right] {\small $E_{7}$};
	\draw[red,  line width = 1pt] (0.7,2.6) -- (1.4,3.3) node [pos=0,below] {\small $E_{4}$};	
	\draw[red,  line width = 1pt] (-1.3,1.7)--(-0.6,2.4) node [pos=1,right] {\small $E_{5}$};
	\draw[red,  line width = 1pt] (0.7,-0.3) -- (1.4,0.4) node [pos=1,above] {\small $E_{8}$};	
	\end{scope}
	\end{tikzpicture}
	\end{center}
	\caption{The model type (b) Sakai surface for the q-$\dPain{E_{3}^{(1)}/A_{5}^{(1)}}$ family} 
	\label{fig:KNY-b-pt-conf}
\end{figure}


\subsubsection{The extended affine Weyl symmetry group} 
\label{ssub:birrep-b}

Since the geometric realization of our surface has changed, the birational representation will change as well. 
In particular, even though the action of the automorphisms $\sigma_{i}$ on the symmetry and surface roots is the same as on 
Figure~\ref{fig:automs}, its realization on $\operatorname{Pic}(\mathcal{X})$ is different,
\begin{align*}
	\sigma_{0} &= w_{\mathcal{E}_{5} - \mathcal{E}_{7}} \circ 	w_{\mathcal{E}_{4} - \mathcal{E}_{8}} \circ 	
		w_{\mathcal{E}_{1} - \mathcal{E}_{3}} \circ 
		w_{\mathcal{E}_{2} - \mathcal{E}_{6}} \circ 
		w_{\mathcal{H}_{f} - \mathcal{E}_{2} - \mathcal{E}_{6}},	\\
	\sigma_{1} &= w_{\mathcal{E}_{5} - \mathcal{E}_{8}} \circ 	
		w_{\mathcal{E}_{1} - \mathcal{E}_{4}} \circ 
		w_{\mathcal{H}_{f} - \mathcal{H}_{g}}.	
\end{align*}

The birational representation of $\widetilde{W}(E_{3}^{(1)})=\operatorname{Aut}(E_{3}^{(1)})\ltimes W(E_{3}^{(1)})$ on the 
(b)-family of surfaces is described in the next Lemma.


\begin{lemma}\label{lem:bir-repE3-b}
	Generators of the extended affine Weyl group $\widetilde{W}(E_{3}^{(1)})$ transform an initial point configuration
	in the (b)-model, expressed in both the KNY parameters and the root variables as
	\begin{equation}\label{eq:base-pt-b}
    \left(\begin{matrix}
    \kappa_{1} & \nu_{1} & \nu_{3} & \nu_{5} & \nu_{7}\\
    \kappa_{2} & \nu_{2} & \nu_{4} & \nu_{6} & \nu_{8}
    \end{matrix}\ ;\ \begin{matrix} f\\ g \end{matrix}\right) \sim
	\left(\begin{matrix}
    a_{0} & a_{1} & a_{2} \\
    a_{3} & a_{4}
    \end{matrix}\ ;\ \begin{matrix} f\\ g \end{matrix}\right), 
	\end{equation}	
	by changing the parameters and the affine coordinates $(f,g)$ as follows:
	\begin{align*}
	w_0&:
		\left(\begin{matrix}
		\frac{\kappa_{1}^{2} \kappa_{2}}{\nu_{2}\nu_{3}\nu_{6}\nu_{7}} & \nu_{1} & \nu_{3} & \nu_{5} & \nu_{7}\\[5pt]
		\frac{\kappa_{1}\kappa_{2}^{2}}{\nu_{2}\nu_{3}\nu_{6}\nu_{7}} & \nu_{2}\left( \frac{\kappa_{1} \kappa_{2}}{\nu_{2}\nu_{3}\nu_{6}\nu_{7}} \right)^{2} & 
		\nu_{4} & \nu_{6}\left( \frac{\kappa_{1} \kappa_{2}}{\nu_{2}\nu_{3}\nu_{6}\nu_{7}} \right)^{2} & \nu_{8}
		\end{matrix}\ ;\ \begin{matrix} 
		\frac{\kappa _1 \kappa _2 f\left(f+\nu _2 \nu _3 g\right)}{\nu _2 \nu _3 \left( \nu _6 \nu _7 f +  \kappa _1 \kappa _2 g\right)}\\[5pt]
		\frac{\nu_{6}\nu_{7}g(f + \nu_{2}\nu_{3} g)}{\nu_{6} \nu_{7} f + \kappa_{1} \kappa_{2} g}\end{matrix} \right) \\
	&\qquad \sim 	
		\left(\begin{matrix}
		\frac{1}{a_{0}} & a_{0}a_{1} & a_{0}a_{2} \\[5pt]
		a_{3} & a_{4}
		\end{matrix}\ ;\ \begin{matrix} \frac{f(a_{0} f + a_{4} g)}{f + a_{4} g} & \frac{g(a_{0} f + a_{4} g)}{a_{0}(f + a_{4} g)}\end{matrix}\right), \\[7pt]
	w_{1}&:
	    \left(\begin{matrix}
	   \frac{(\nu_{4}\nu_{8})^{2}}{\kappa_{1}} & \nu_{1} & \nu_{3} & \nu_{5} & \nu_{7}\\[5pt]
	    \frac{\kappa_{1} \kappa_{2}}{\nu_{4}\nu_{8}} & \frac{\nu_{2} \nu_{4} \nu_{8}}{\kappa_{1}} & \nu_{4} & \frac{\nu_{4} \nu_{6} \nu_{8}}{\kappa_{1}} & \nu_{8}
	    \end{matrix}\ ;\ 
		\begin{matrix} 		\frac{\nu_{4} \nu_{8} f}{\kappa_{1}}\\[5pt]
		\frac{\nu_{8}(f-\nu_{4})g}{\nu_{8} f - \kappa_{1}}\end{matrix} \right) \sim 			
		\left(\begin{matrix}
	    a_{0}a_{1} & \frac{1}{a_{1}} & a_{1}a_{2} \\[5pt]
	    a_{3} & a_{4}
	    \end{matrix}\ ;\ \begin{matrix} \frac{f}{a_{1}}\\[5pt] 
		\frac{(f-1)g}{(f-a_{1})} \end{matrix}\right), 	\\[7pt]
	w_{2}&:
	    \left(\begin{matrix}
	    \frac{\kappa_{1}\kappa_{2}}{\nu_{1}\nu_{5}} & \nu_{1} & \nu_{3} & \nu_{5} & \nu_{7}\\[5pt]
	    \frac{(\nu_{1}\nu_{5})^{2}}{\kappa_{2}} & \frac{\nu_{1}\nu_{2}\nu_{5}}{\kappa_{2}} & \nu_{4} & \frac{\nu_{1} \nu_{5} \nu_{6}}{\kappa_{2}} & \nu_{8}
	    \end{matrix}\ ;\ \begin{matrix} \frac{\kappa_{2}f(\nu_{1}g - 1)}{\nu_{1}(\kappa_{2}g-\nu_{5})}\\[5pt] \frac{\kappa_{2} g}{\nu_{1} \nu_{5}} \end{matrix}\right) \sim
		\left(\begin{matrix}
	    a_{0}a_{2} & a_{1}a_{2} & \frac{1}{a_{2}} \\[5pt]
	    a_{3} & a_{4}
	    \end{matrix}\ ;\ \begin{matrix} \frac{a_{2} f (g-1)}{a_{2} g - 1}\\[5pt] a_{2} g \end{matrix}\right), 	\\[7pt]	
	w_{3}&:
	    \left(\begin{matrix}
	    \kappa_{1} & \nu_{1} & \nu_{3} & \nu_{5} & \nu_{7}\\[5pt]
	    \kappa_{2} & \nu_{2}\left(\frac{\kappa_{1}\kappa_{2}}{\nu_{2}\nu_{3}\nu_{5}\nu_{8}}\right)^{2} & \nu_{4} & 
			\nu_{6} \left(\frac{\nu_{2}\nu_{3}\nu_{5}\nu_{8}}{\kappa_{1}\kappa_{2}}\right)^{2} & \nu_{8}
	    \end{matrix}\ ;  \begin{matrix} 
		\frac{\kappa_{1}\kappa_{2}f(\nu_{8}(f + \nu_{2}\nu_{3}g)-\kappa_{1})}{\nu_{2}\nu_{3}\nu_{8}(\nu_{5}\nu_{8}f + \kappa_{1}(\kappa_{2}g - \nu_{5}))}
		\\[5pt] 
		\frac{\nu_{5}\nu_{8} g (\kappa_{2} f + \nu_{2} \nu_{3} (\kappa_{2} g - \nu_{5}))}{\kappa_{2}(\nu_{5} \nu_{8} f + \kappa_{1} (\kappa_{2} g - \nu_{5}))}
		\end{matrix}\right) \\ 
		&\qquad \sim 
		\left(\begin{matrix}
	    a_{0} & a_{1} & a_{2} \\[5pt]
	    \frac{1}{a_{3}} & a_{3}^{2}a_{4}
	    \end{matrix}\ ;\ \begin{matrix} 
		\frac{a_{1}a_{2}f (a_{0}(f-a_{1}) + a_{4} g)}{a_{4}(f + a_{1}(a_{2} g - 1))}\\[5pt] 
		\frac{g(a_{0} a_{2} f + a_{4}(a_{2} g - 1))}{a_{0}a_{2}(f+ a_{1}(a_{2} g - 1))} \end{matrix}\right), 	\\[7pt]	
	w_{4}&:
	    \left(\begin{matrix}
	    \kappa_{1} & \nu_{1} & \nu_{3} & \nu_{5} & \nu_{7}\\[5pt]
	    \kappa_{2} & \nu_{2} \left(\frac{\nu_{1}\nu_{4}\nu_{6}\nu_{7}}{\kappa_{1}\kappa_{2}}\right)^{2} & \nu_{4} & \nu_{6} 			\left(\frac{\kappa_{1}\kappa_{2}}{\nu_{1}\nu_{4}\nu_{6}\nu_{7}}\right)^{2} & \nu_{8}
	    \end{matrix}\ ;\ \begin{matrix} 
		\frac{\nu_{4}\nu_{6}\nu_{7}f(f - \nu_{1}g(f - \nu_{4}))}{\nu_{4}\nu_{6}\nu_{7}f- \kappa_{1}\kappa_{2}g(f-\nu_{4})}\\[5pt] 
		\frac{\kappa_{1}\kappa_{2} g (f - \nu_{1}g(f - \nu_{4}))}{\nu_{1}\nu_{4}(f(1 - \nu_{1} g)+ \kappa_{1}\kappa_{2} g)}
			\end{matrix}\right) \\
		&\qquad \sim
		\left(\begin{matrix}
	    a_{0} & a_{1} & a_{2} \\[5pt]
	    a_{3}a_{4}^{2} &\frac{1}{a_{4}} 
	    \end{matrix}\ ;\ \begin{matrix} 
		\frac{f(f + g - fg )}{f - a_{4}(f-1) g}\\[5pt] \frac{a_{4}g(f g - f - g)}{f(g-1) - a_{4}g} \end{matrix}\right),  \\[7pt]
	\sigma_{0}&:
	    \left(\begin{matrix}
	    \frac{(\nu_{4}\nu_{8})^{2}}{\kappa_{1}} & \nu_{1} & \nu_{3} & \nu_{5} & \nu_{7}\\[5pt]
	    \frac{\nu_{1}\nu_{2}\nu_{3}\nu_{5}\nu_{6} \nu_{7}}{\kappa_{1} \kappa_{2}} & \frac{\nu_{2} \nu_{4} \nu_{8}}{\kappa_{1}} & 
			\nu_{4} & \frac{\nu_{4}\nu_{6}\nu_{8}}{\kappa_{1}} & \nu_{8}
	    \end{matrix}\ ;\ \begin{matrix} \frac{\nu_{4}\nu_{8}f}{\kappa_{1}}\\[5pt] 
			-\frac{f}{\nu_{1}\nu_{2}\nu_{3}g} \end{matrix}\right) \sim
		\left(\begin{matrix}
	    \frac{1}{a_{2}} & \frac{1}{a_{1}} & \frac{1}{a_{0}} \\[5pt]
	    \frac{1}{a_{4}} & \frac{1}{a_{3}}
	    \end{matrix}\ ;\ \begin{matrix} \frac{f}{a_{1}}\\[5pt] -\frac{a_{0}f}{a_{4}g} \end{matrix}\right), 	\\[7pt]	
	\sigma_{1}&:
	    \left(\begin{matrix}
	    \frac{\nu_{1}\nu_{4}\nu_{5}\nu_{8}}{\kappa_{2}} & \nu_{1} & \nu_{3} & \nu_{5} & \nu_{7}\\[5pt]
	    \frac{\nu_{1}\nu_{4}\nu_{5}\nu_{8}}{\kappa_{1}} & \frac{(\nu_{1}\nu_{4})^{2}}{\nu_{2}\nu_{3}^{2}} & \nu_{4} & \frac{(\nu_{5}\nu_{8})^{2}}{\nu_{6}\nu_{7}^{2}} & \nu_{8}
	    \end{matrix}\ ;\ \begin{matrix} \nu_{1}\nu_{4}g\\[5pt] 
			\frac{f}{\nu_{1}\nu_{4}} \end{matrix}\right) \sim
		\left(\begin{matrix}
	    \frac{1}{a_{0}} & \frac{1}{a_{2}} & \frac{1}{a_{1}} \\[5pt]
	    \frac{1}{a_{3}} & \frac{1}{a_{4}}
	    \end{matrix}\ ;\ \begin{matrix} g\\[5pt] f \end{matrix}\right), 	\\[7pt]				
	\sigma_{2}&:
	    \left(\begin{matrix}
	   	\frac{\kappa_{1}\kappa_{2}\nu_{4}\nu_{8}}{\nu_{2}\nu_{3}\nu_{6}\nu_{7}} & \nu_{1} & \nu_{3} & \nu_{5} & \nu_{7}\\[5pt]
	    \frac{\kappa_{1}\nu_{1}\nu_{5}}{\nu_{4}\nu_{8}} & \frac{\kappa_{1}\nu_{1}^{2}\nu_{4}}{\nu_{2} \nu_{3}^{2} \nu_{8}} & \nu_{4} & 
			\frac{\kappa_{1} \nu_{5}^{2} \nu_{8}}{\nu_{4}\nu_{6}\nu_{7}^{2}} & \nu_{8}
	    \end{matrix}\ ;\ \begin{matrix} -\frac{\nu_{4}f}{\nu_{2}\nu_{3}g}\\[5pt] 
			\frac{\nu_{8}f}{\kappa_{1}\nu_{1}}\end{matrix}\right)
		\sim
		\left(\begin{matrix}
	    a_{2} & a_{0} & a_{1} \\[5pt]
	    a_{4} & a_{3}
	    \end{matrix}\ ;\ \begin{matrix} -\frac{a_{0}f}{a_{4}g}\\[5pt] \frac{f}{a_{1}} 		\end{matrix}\right).				
	\end{align*}
\end{lemma}

The proof of this Lemma is similar to Lemma~\ref{lem:bir-repE3-a}. Note that we use the following root variable parameterization, 
obtained from \eqref{eq:root-vars-KNY-b}:
\begin{equation*}
	\kappa_{1} = a_{1}\, \nu_{4}\nu_{8},\quad \kappa_{2} = a_{2} \, \nu_{1}\nu_{5},\quad 
	\nu_{2} = \frac{a_{4}}{a_{0}}\, \frac{\nu_{1}\nu_{4}}{\nu_{3}},\quad  \nu_{6} = \frac{a_{1}a_{2}}{a_{4}}\, \frac{\nu_{5} \nu_{8}}{\nu_{7}}.
\end{equation*}


\subsubsection{The discrete Painlev\'e equation $[01\overline{1}00]$ on the q-$\dPain{E_{3}^{(1)}/A_{5}^{(1)};b}$ Surface Family} 
\label{ssub:dP-b}

The evolution of parameters \eqref{eq:KNY-par-evol}
induces the following evolution on the root variables \eqref{eq:root-vars-KNY-b}:
\begin{equation*}
	\overline{a}_{0} = a_{0},\quad \overline{a}_{1} = \frac{a_{1}}{q}, \quad \overline{a}_{2} = q a_{2},\quad \overline{a}_{3} = a_{3},\quad \overline{a}_{4} = a_{4},
\end{equation*}
that in turn induces the translation \eqref{eq:dP-A5-b-trans} on the symmetry roots,
\begin{equation*}
\phi_*:\upalpha=\langle\alpha_0,\alpha_1,\alpha_2;\alpha_3,\alpha_{4}\rangle\mapsto
\phi_*(\upalpha)=\upalpha + \langle 0,1,-1;0,0 \rangle\delta.
\end{equation*}
This translation can be expressed in terms of the generators as 
$\phi = \sigma_{2}^{4}\circ w_{0} \circ w_{2} = w_{1} \circ w_{0} \circ \sigma_{2}^{4}$, resulting in the mapping 
\begin{equation*}
	\phi(f,g) = \left( - \frac{\nu_{1}\nu_{2} \nu_{3} \nu_{4}g(\kappa_{2} g - \nu_{5})}{\kappa_{2} f (\nu_{1}g - 1)},
	\frac{\nu_{1}\nu_{2}\nu_{3}\nu_{4}(\kappa_{2}g - \nu_{5})(\nu_{5}\nu_{6}\nu_{7}f(\nu_{1}g-1) + \kappa_{1}\kappa_{2}g(\kappa_{2}g - \nu_{5}))}{
	\kappa_{1}\kappa_{2}^{2}f(\nu_{1}g - 1)(\kappa_{2}f(\nu_{1}g - 1)+ \nu_{1}\nu_{2}\nu_{3}g(\kappa_{2}g - \nu_{5}))}\right),
\end{equation*}
which is equivalent to \cite[(8.14)]{KajNouYam:2017:GAPE},
\begin{equation*}
	f \overline{f} = - \nu_{2} \nu_{3} \nu_{4}\frac{g\left(g - \frac{\nu_{5}}{\kappa_{2}}\right)}{g-\frac{1}{\nu_{1}}},\quad 
	g \underline{g} = -\frac{1}{\nu_{1}\nu_{2}\nu_{3}}\, \frac{f\left(f - \frac{\kappa_{1}}{\nu_{8}}\right)}{f - \nu_{4}},
\end{equation*}
see also \cite[(2.15--2.16)]{Sak:2007:PDPETLF} and the $q$-$\Pain{III}\to\Pain{III}$ degeneration in \cite{Sak:2001:RSAWARSGPE}.


\section{The Identification Procedure} 
\label{sec:identification}

In this section we follow the procedure introduced in  \cite{DzhFilSto:2020:RCDOPWHWDPE} to formally identify recurrence 
\eqref{eq:qxn-back-for} as a discrete Painlev\'e equation. This procedure consists of several steps that can be grouped
into two parts. The first part is geometric, where we first resolve the singularities of the mapping, linearize the mapping
on the Picard lattice of the resulting algebraic surface, determine the decomposition of the unique anti-canonical divisor
into irreducible components whose intersection configuration is described by an affine Dynkin diagram. The type of this
diagram is known as the surface type of the equation in the Sakai classification scheme. 
We do this part in Section~\ref{sub:surface_type}.

The second part is algebraic, where we do some preliminary change of basis between the Picard lattices of our surface and
of one of the model examples so that the surface roots match. This allows us to obtain the expressions for the standard 
symmetry roots in terms of the basis of the Picard lattice of our $q$-Laguerre surface. Using the evolution on the Picard 
lattice we can determine the translation direction and see if it is conjugate to one of the known examples. If so, we can 
do that conjugation to find the final correspondence between the bases of these two Picard lattices and then determine the 
underlying change of coordinates. We do this part in Section~\ref{sub:matching}, except that in our case the $q$-Laguerre dynamic is
conjugated to a composition of two standard mappings.  

Since the procedure we follow here is by now quite standard, we only outline the computations and refer the interested reader to 
\cite{DzhFilSto:2020:RCDOPWHWDPE} and \cite{KajNouYam:2001:SFQE} for details.

\subsection{Resolving the Singularities, Identifying the Surface Type, 
and Linearizing the Dynamics on the Picard Lattice} 
\label{sub:surface_type}

We begin by regularizing the mapping given by recurrence~\eqref{eq:qxn-back-for}. Note that this recurrence 
defines two half-step mappings, the forward half-step $\varphi_{1}^{(n)}: (x_{n},y_{n})\mapsto (x_{n},y_{n+1})$, where
\begin{equation}\label{eq:y-fwd}
	y_{n+1} = \frac{q^{2n+1}Q(T-x_{n})(x_{n}-1)+x_{n}(x_{n}y_{n}-1)}{x_{n}^2(x_{n}y_{n}-1)},
\end{equation}
and the backward half-step $\varphi_{2}^{(n)}:(x_{n},y_{n})\mapsto (x_{n-1},y_{n})$, where
\begin{equation}\label{eq:x-back}
	x_{n-1} = \frac{(y - q^{n})(x_{n} y_{n} -1) - q^{2n} Q (T y_{n} - 1) (y_{n}-1)}{y_{n}(y_{n}-q^{n})(x_{n}y_{n}-1)}.
\end{equation}
Each mapping is clearly birational, and so after inverting one of them we obtain the 
full forward and backward mappings $\varphi^{(n)} = \left(\varphi_{2}^{(n+1)}\right)^{-1} \circ \varphi_{1}^{(n)}: (x_{n},y_{n})\mapsto (x_{n+1},y_{n+1})$
and $(\varphi^{(n)})^{-1} = \left(\varphi_{1}^{(n-1)}\right)^{-1} \circ \varphi_{2}^{(n)}$. The explicit expressions for these
mappings are complicated and we omit them.

Extending these mappings to the compactification $\mathbb{P}^{1} \times \mathbb{P}^{1}$ and using the standard techniques described, e.g., in
\cite{DzhFilSto:2020:RCDOPWHWDPE} or \cite{HuDzhChe:2020:PLUEDPE}, we get the following base points of the mappings 
shown on Figure~\ref{fig:qLaguerre-pt-conf} (left)
\begin{equation}\label{eq:base-pt-qlaguerre}
	\begin{aligned}
		&q_{1}(x=1,y=1),\quad q_{2}\left(x=T,y=\frac{1}{T}\right),\\
		&q_{3}\left(x=0,\frac{1}{y}= 0\right)\leftarrow q_{4}\left(x= 0,\frac{1}{xy} = 0\right)\leftarrow 
		q_{5}\left(x = 0,\frac{1}{x^2 y} = \frac{q^{-2n-1}}{QT}\right),\\
		&q_{6}\left(\frac{1}{x} = 0, y= 0\right)\leftarrow q_{7}\left(\frac{1}{x} = 0, xy = 1-q^n Q\right),\quad 
		q_{8}\left(\frac{1}{x}=0,y=q^n\right).
	\end{aligned}
\end{equation}
Note that the points $q_{1}$, $q_{2}$, $q_{3}$, and $q_{6}$ lie on the $(1,1)$ curve given by the equation $xy-1=0$ in the affine $(x,y)$-chart.
Next, modifying the geometry by applying the blowup procedure at the base points, we extend each mapping from a birational transformation of $\mathbb{C}\times\mathbb{C}$ to an isomorphism on the family of Sakai 
surfaces parameterized by $q$, $Q$, $T$, and $n$ and described schematically on Figure~\ref{fig:qLaguerre-pt-conf} (right).

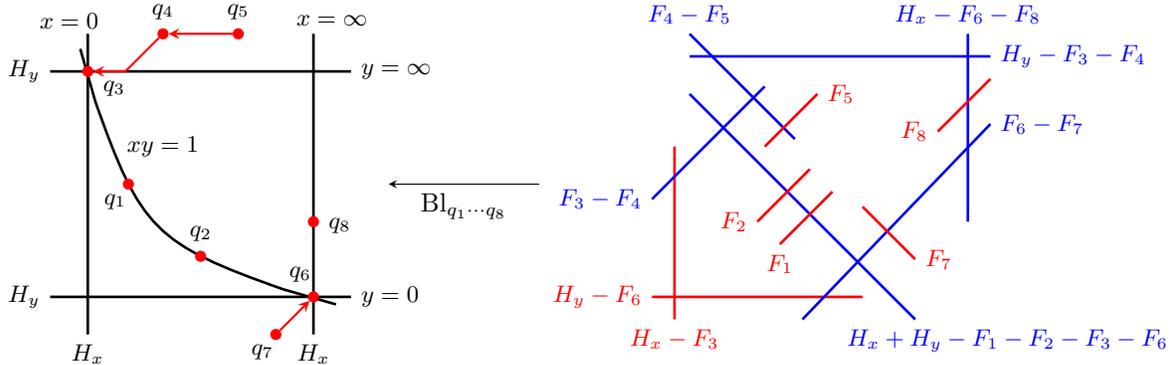
\begin{figure}[ht]
	\begin{center}		
	\begin{tikzpicture}[>=stealth,basept/.style={circle, draw=red!100, fill=red!100, thick, inner sep=0pt,minimum size=1.2mm}]
	\begin{scope}[xshift=0cm,yshift=0cm]
	\draw [black, line width = 1pt] (-0.5,0) -- (3.5,0)	node [pos=0,left] {\small $H_{y}$} node [pos=1,right] {\small $y=0$};
	\draw [black, line width = 1pt] (-0.5,3) -- (3.5,3) node [pos=0,left] {\small $H_{y}$} node [pos=1,right] {\small $y=\infty$};
	\draw [black, line width = 1pt] (0,-0.5) -- (0,3.5) node [pos=0,below] {\small $H_{x}$} node [pos=1,xshift = -7pt, yshift=5pt] {\small $x=0$};
	\draw [black, line width = 1pt] (3,-0.5) -- (3,3.5) node [pos=0,below] {\small $H_{x}$} node [pos=1,xshift = 7pt, yshift=5pt] {\small $x=\infty$};
	\draw [black, line width = 1pt] (-0.1,3.3) -- (0.1,2.6)  .. controls (0.7,0.9) and (0.9,0.7) .. (2.6,0.1)--(3.3,-0.1);
	\node at (1,2) {\small $x y = 1$};
	\node (q1) at (0.54,1.5) [basept,label={[xshift = -5pt, yshift=-15pt] \small $q_{1}$}] {};
	\node (q2) at (1.5,0.54) [basept,label={[xshift = 0pt, yshift=0pt] \small $q_{2}$}] {};
	\node (q3) at (0,3) [basept,label={[xshift = 10pt, yshift=-15pt] \small $q_{3}$}] {};
	\node (q4) at (1,3.5) [basept,label={[xshift = 0pt, yshift=0pt] \small $q_{4}$}] {};
	\node (q5) at (2,3.5) [basept,label={[xshift = 0pt, yshift=0pt] \small $q_{5}$}] {};
	\node (q6) at (3,0) [basept,label={[xshift=-5pt,yshift=0pt] \small $q_{6}$}] {};
	\node (q7) at (2.5,-0.5) [basept,label={[xshift=-5pt,yshift=-15pt] \small $q_{7}$}] {};
	\node (q8) at (3,1) [basept,label={[xshift=10pt,yshift=-10pt] \small $q_{8}$}] {};
	\draw [red, line width = 0.8pt, ->] (q4) -- (0.5,3) -- (q3);
	\draw [red, line width = 0.8pt, ->] (q5) -- (q4);
	\draw [red, line width = 0.8pt, ->] (q7) -- (q6);
	\end{scope}
	\draw [->] (6,1.5)--(4,1.5) node[pos=0.5, below] {$\operatorname{Bl}_{q_{1}\cdots q_{8}}$};
	\begin{scope}[xshift=8.5cm,yshift=0cm]
	\draw[red, line width = 1pt] (-1,0) -- (1.8,0) node [pos=0,left] {\small $H_{y}-F_{6}$};
	\draw[blue, line width = 1pt] (-0.5,3.2) -- (3.5,3.2) node [pos=1,right] {\small $H_{y}-F_{3} - F_{4}$};
	\draw[red, line width = 1pt] (-0.7,-0.3) -- (-0.7,2) node [pos=0,below] {\small $H_{x}-F_{3}$};
	\draw[blue, line width = 1pt] (3.2,1) -- (3.2,3.5) node [pos=1,above] {\small $H_{x}-F_{6}-F_{8}$};
	\draw[blue, line width = 1pt] (1,-0.3) -- (3.5,2.3) node [pos=1,right] {\small $F_{6}-F_{7}$};
	\draw[blue, line width = 1pt] (-1,1.3) -- (0.5,2.8) node [pos=0,left] {\small $F_{3}-F_{4}$};
	\draw[blue, line width = 1pt] (-0.5,3.5) -- (0.9,2.1) node [pos=0,above] {\small $F_{4}-F_{5}$};
	\draw[blue, line width = 1pt] (-0.5,2.7) -- (2.5,-0.3) node [pos=1,xshift = 35pt,yshift = -8 pt] {\small $H_{x} + H_{y} - F_{1} - F_{2} - F_{3}-F_{6}$};
	\draw[red, line width = 1pt] (0.7,0.7) -- (1.4,1.4) node [pos=0,below] {\small $F_{1}$};
	\draw[red, line width = 1pt] (0.4,1.0) -- (1.1,1.7) node [pos=0,left] {\small $F_{2}$};
	\draw[red, line width = 1pt] (0.5,2) -- (1.2,2.7) node [pos=1,right] {\small $F_{5}$};
	\draw[red, line width = 1pt] (2.8,2.2) -- (3.5,2.9) node [pos=0,left] {\small $F_{8}$};
	\draw[red, line width = 1pt] (1.8,1.2) -- (2.5,0.5) node [pos=1,right] {\small $F_{7}$};
	\end{scope}
	\end{tikzpicture}
	\end{center}
	\caption{The Sakai surface for the $q$-Laguerre recurrence} 
	\label{fig:qLaguerre-pt-conf}
\end{figure}	

Recall that the key object in the geometric approach to Painlev\'e equations is the \emph{Picard lattice} of this surface family. It is well-known that 
the Picard lattice of $\mathbb{P}^{1} \times \mathbb{P}^{1}$ is generated by the classes of the coordinate lines, 
$\operatorname{Pic}(\mathbb{P}^{1} \times \mathbb{P}^{1}) = \operatorname{Span}_{\mathbb{Z}}\{\mathcal{H}_{x},\mathcal{H}_{y}\}$.  
Each blowup is a surgery that adds a class of the exceptional divisor $\mathcal{F}_{i} = [F_{i}]$ to the Picard lattice and so for each 
surface  $\mathcal{X}$ in the family we get 
$\operatorname{Pic}(\mathcal{X}) = \operatorname{Span}_{\mathbb{Z}}\{\mathcal{H}_{x},\mathcal{H}_{y},\mathcal{F}_{1},\ldots,\mathcal{F}_{8}\}$.
This lattice is equipped with the \emph{intersection product} defined on the generators by 
$\mathcal{H}_{x}\bullet \mathcal{H}_{x} = \mathcal{H}_{y}\bullet \mathcal{H}_{y} = \mathcal{H}_{x}\bullet \mathcal{F}_{i} =
\mathcal{H}_{y}\bullet \mathcal{F}_{j} = 0$, $\mathcal{H}_{x}\bullet \mathcal{H}_{y} = 1$, and 
$\mathcal{F}_{i}\bullet \mathcal{F}_{j} = - \delta_{ij}$. Using this inner product we can assign to each curve on $\mathcal{X}$ 
its \emph{self-intersection} index, e.g., all exceptional curves $F_{i}$ have index $-1$; such curves are marked in red on 
Figure~\ref{fig:qLaguerre-pt-conf} (right). But there are other curves with self-intersection index $-1$, in particular, proper transforms
$\mathcal{H}_{i}-F_{j}$ of the coordinate lines passing through the blowup points. There are infinitely many $-1$ curves, but curves with higher
negative self-intersection index are more special. Usually, as is in our case, there are only finitely many curves with self-intersection 
index $-2$ that are the \emph{irreducible components} of the \emph{unique anti-canonical divisor class of canonical type} (the generalized
Halphen surface condition, see \cite{Sak:2001:RSAWARSGPE}). Such curves are marked in blue on Figure~\ref{fig:qLaguerre-pt-conf} (right)
and their intersection configuration is described by an affine Dynkin diagram which in our case is $A_{5}^{(1)}$, see 
Figure~\ref{fig:surface-roots-qlaguerre}. This is the \emph{surface type} of our recurrence. 
\begin{figure}[ht]
\begin{equation}\label{eq:d-roots-lw}
	\raisebox{-45pt}{\begin{tikzpicture}[
			elt/.style={circle,draw=black!100,thick, inner sep=0pt,minimum size=2mm}]
		\path 				(0,1)	node 	(d0) 	[elt, label=above:{$\delta_{0}$} ] {}
		        (-{sqrt(3)/2},1/2)	node 	(d1) 	[elt, label=left:{$\delta_{1}$} ] {}
		        (-{sqrt(3)/2},-1/2) node  	(d2)	[elt, label=left:{$\delta_{2}$} ] {}
		     	   			(0,-1)	node  	(d3) 	[elt, label=below:{$\delta_{3}$} ] {}
		        ({sqrt(3)/2},-1/2) 	node  	(d4) 	[elt, label=right:{$\delta_{4}$} ] {}
		        ({sqrt(3)/2},1/2) 	node 	(d5) 	[elt, label=right:{$\delta_{5}$} ] {};
		\draw [black,line width=1pt ] (d0) -- (d1) -- (d2) -- (d3) -- (d4) -- (d5)--(d0);
	\end{tikzpicture}}\qquad 
			\begin{alignedat}{2}
			\delta_{0} &= \mathcal{F}_{3} - \mathcal{F}_{4},&\quad  
			\delta_{3} &=\mathcal{H}_{x}-\mathcal{F}_{6} - \mathcal{F}_{8},\\
			\delta_{1} &=  \mathcal{F}_{4} - \mathcal{F}_{5}, &\quad  \delta_{4} &= \mathcal{F}_{6} - \mathcal{F}_{7},\\
			\delta_{2} &= \mathcal{H}_{y}- \mathcal{F}_{3} - \mathcal{F}_{4}, &\quad  \delta_{5} &=  \mathcal{H}_{x}+\mathcal{H}_{y} - \mathcal{F}_{1} - \mathcal{F}_{2} - \mathcal{F}_{3} - \mathcal{F}_{6}.
			\end{alignedat}
\end{equation}
	\caption{The initial surface root basis for the $q$-Laguerre recurrence}
	\label{fig:surface-roots-qlaguerre}
\end{figure}
 
We can also see the action of the dynamic on the standard basis of the Picard lattice. This is a direct computation summarized in the following Lemma. 

\begin{lemma}\label{lem:q-Lag-Pica-action}
	The induced linear action of the full mapping ${\varphi_{*}}$ on the Picard lattice is given by
	\begin{alignat*}{2}
		\mathcal{H}_{x}& \mapsto \overline{\mathcal{H}}_{x}+ 3\overline{\mathcal{H}}_{y}- \overline{\mathcal{F}}_{123678}, \qquad&
		\mathcal{H}_{y}& \mapsto 3\overline{\mathcal{H}}_{x} + 6\overline{\mathcal{H}}_{y} - 2\overline{\mathcal{F}}_{1237}- \overline{\mathcal{F}}_{45}-3\overline{\mathcal{F}}_{68},  \\
		\mathcal{F}_{1} &\mapsto \overline{\mathcal{H}}_{x}+2\overline{\mathcal{H}}_{y} -\overline{\mathcal{F}}_{23678}, \qquad&
		\mathcal{F}_{2} &\mapsto \overline{\mathcal{H}}_{x} +2\overline{\mathcal{H}}_{y} -\overline{\mathcal{F}}_{13678},\\
		\mathcal{F}_{3} &\mapsto \overline{\mathcal{H}}_{x} +3\overline{\mathcal{H}}_{y} -\overline{\mathcal{F}}_{1235678},\qquad&
		\mathcal{F}_{4} &\mapsto \overline{\mathcal{H}}_{x} +3\overline{\mathcal{H}}_{y} -\overline{\mathcal{F}}_{1234678},\\
		\mathcal{F}_{5} &\mapsto \overline{\mathcal{H}}_{x} +2\overline{\mathcal{H}}_{y} -\overline{\mathcal{F}}_{12678},\qquad&
		\mathcal{F}_{6} &\mapsto \overline{\mathcal{H}}_{x} +2\overline{\mathcal{H}}_{y} -\overline{\mathcal{F}}_{12368}, \\
		\mathcal{F}_{7} &\mapsto \overline{\mathcal{H}}_{y}-\overline{\mathcal{F}}_{8}, \qquad&
		\mathcal{F}_{8} &\mapsto \overline{\mathcal{H}}_{y}-\overline{\mathcal{F}}_{6},
	\end{alignat*}
	where we use the notation $\mathcal{F}_{i\cdots j} = \mathcal{F}_{i} + \cdots \mathcal{F}_{j}$.	
	The evolution of parameters (hence, the base points) is as expected,
	$\mathbf{b}=\{T, Q, q, n\}\mapsto \overline{\mathbf{b}}=\{T, Q, q, n+1\}$.
\end{lemma}

The next step is to see if we can match this equation to one of the standard examples. We do it in the next section.

\subsection{Matching the Geometry to a Standard Model and Identifying the Equation} 
\label{sub:matching}
As we have seen in Section~\ref{sec:dP-A5-surf-std}, there are two standard geometric realizations of this surface given in 
\cite{KajNouYam:2017:GAPE}, termed \emph{(a)-model} and \emph{(b)-model}. Both realizations are birationally equivalent, but given that the birational 
representation using the (b)-model, given in Lemma~\ref{lem:bir-repE3-b}, is somewhat simpler than the one for the (a)-model given in Lemma~\ref{lem:bir-repE3-a},
we match our geometry to the one of (b)-model described in Section~\ref{ssub:points-b}. This is done in three steps. First, we make 
some initial change of basis between the Picard lattices for our surface and that of a standard example. This can be done by matching the 
surface roots on Figure~\ref{eq:d-roots-lw} with the one for the (b)-model shown on the right on Figure~\ref{eq:d-roots-KNY-b}.
There are of course many different ways to do so, but at this point it is enough to find \emph{some} identification that we can later adjust using
diagram automorphism elements from $\widetilde{W}(E_{3}^{(1)})$. For example, we can just match the nodes with the same indices on these two diagrams, 
\begin{alignat*}{2}
	\delta_{0} &= \mathcal{F}_{3} - \mathcal{F}_{4} = \mathcal{H}_{f} - \mathcal{E}_{1} - \mathcal{E}_{2},&\qquad  
	\delta_{3} &=\mathcal{H}_{x}-\mathcal{F}_{6} - \mathcal{F}_{8} =  \mathcal{H}_{f}- \mathcal{E}_{5} - \mathcal{E}_{6}, \\
	\delta_{1} &=  \mathcal{F}_{4} - \mathcal{F}_{5} = \mathcal{E}_{2} - \mathcal{E}_{3}, &\qquad  
	\delta_{4} &= \mathcal{F}_{6} - \mathcal{F}_{7} = \mathcal{E}_{6}-\mathcal{E}_{7},\\
	\delta_{2} &= \mathcal{H}_{y}- \mathcal{F}_{3} - \mathcal{F}_{4} = \mathcal{H}_{g} - \mathcal{E}_{2} - \mathcal{E}_{4}, &\qquad  
	\delta_{5} &=  \mathcal{H}_{x}+\mathcal{H}_{y} - \mathcal{F}_{1} - \mathcal{F}_{2} - \mathcal{F}_{3} - \mathcal{F}_{6} = 
	\mathcal{H}_{g} - \mathcal{E}_{6} - \mathcal{E}_{8},
\end{alignat*}
which can be done via the change of basis 
\begin{equation}
	\begin{aligned}
		\mathcal{H}_{f} &= \mathcal{H}_{x}, &\quad \mathcal{E}_{1} &= \mathcal{H}_{x} - \mathcal{F}_{3}, &\quad \mathcal{E}_{3} &= \mathcal{F}_{5}, &\quad 
		\mathcal{E}_{5} &= \mathcal{F}_{8}, &\quad \mathcal{E}_{7} &= \mathcal{F}_{7},\\ 
		\mathcal{H}_{g} &= \mathcal{H}_{x} + \mathcal{H}_{y} - \mathcal{F}_{2} - \mathcal{F}_{3}, &\quad 
		\mathcal{E}_{2} &= \mathcal{F}_{4}, &\quad  \mathcal{E}_{4} &=  \mathcal{H}_{x} - \mathcal{F}_{2}, &\quad \mathcal{E}_{6} &= \mathcal{F}_{6},
		&\quad \mathcal{E}_{8} &= \mathcal{F}_{1}. 
	\end{aligned}
\end{equation}
Using this identification we get the symmetry roots for the $q$-Laguerre recurrence corresponding to the (b)-model symmetry roots  
shown on the left on Figure~\ref{eq:d-roots-KNY-b},
\begin{equation}
	\begin{aligned}
			\alpha_{0} &= \mathcal{H}_{f} + \mathcal{H}_{g} - \mathcal{E}_{2}-\mathcal{E}_{3} - \mathcal{E}_{6} - \mathcal{E}_{7} &&= 
			2 \mathcal{H}_{x} + \mathcal{H}_{y} - \mathcal{F}_{2} - \mathcal{F}_{3} - \mathcal{F}_{4} - \mathcal{F}_{5} - \mathcal{F}_{6} - \mathcal{F}_{7}	\\		
			\alpha_{1} &= \mathcal{H}_{f}  - \mathcal{E}_{4} - \mathcal{E}_{8} &&= \mathcal{F}_{2} - \mathcal{F}_{1}, \\
			\alpha_{2} &= \mathcal{H}_{g} - \mathcal{E}_{1} - \mathcal{E}_{5} &&=\mathcal{H}_{y} - \mathcal{F}_{2} - \mathcal{F}_{8},  \\
			\alpha_{3} &=\mathcal{H}_{f} + \mathcal{H}_{g} - \mathcal{E}_{2} - \mathcal{E}_{3} - \mathcal{E}_{5} - \mathcal{E}_{8} &&=
			2 \mathcal{H}_{x} + \mathcal{H}_{y} - \mathcal{F}_{1} - \mathcal{F}_{2} - \mathcal{F}_{3} - \mathcal{F}_{4} - \mathcal{F}_{5} - \mathcal{F}_{8},	\\
			\alpha_{4} &= \mathcal{H}_{f} + \mathcal{H}_{g} - \mathcal{E}_{1} - \mathcal{E}_{4} - \mathcal{E}_{6} - \mathcal{E}_{7} &&=
			\mathcal{H}_{y} - \mathcal{F}_{6} - \mathcal{F}_{7}.
	\end{aligned}
\end{equation}
Now we can, using the action of $\varphi_{*}$ on $\operatorname{Pic}(\mathcal{X})$ given in Lemma~\ref{lem:q-Lag-Pica-action}, identify the translation element on 
the symmetry sublattice; 
\begin{equation*}
\varphi_*:\upalpha=\langle\alpha_0,\alpha_1,\alpha_2;\alpha_3,\alpha_{4}\rangle\mapsto
\varphi_*(\upalpha)=\upalpha+ \langle -1,0,1;-1,1 \rangle\delta.
\end{equation*}
From this we immediately see that the $q$-Laguerre dynamic is \emph{not conjugated} to either of the standard examples, since it occurs on both sublattices, but it is almost 
the same as their composition, we just need to adjust the change of basis, which is the second step. 

Consider the automorphism $\sigma_{2}^{2}$, see Figure~\ref{fig:automs}, which acts on the symmetry roots as 
\begin{equation*}
\sigma_{2}^{2}(\langle\alpha_0,\alpha_1,\alpha_2;\alpha_3,\alpha_{4}\rangle) = \langle\alpha_2,\alpha_0,\alpha_1;\alpha_3,\alpha_{4}\rangle,
\end{equation*}
so it transforms the dynamic into $[01\overline{1}\overline{1}1]$, and that is the composition of two standard translations. Acting by this automorphism 
 on our change of basis creates a somewhat less obvious change of basis, 
\begin{equation}\label{eq:change-basis-final}
	\begin{aligned}
		\mathcal{H}_{f} &= \mathcal{H}_{x} + \mathcal{H}_{y} - \mathcal{F}_{2} - \mathcal{F}_{3}, &\qquad 
					\mathcal{H}_{x} &= \mathcal{H}_{f} + \mathcal{H}_{g} - \mathcal{E}_{2} - \mathcal{E}_{6}, \\ 
		\mathcal{H}_{g} &= 2\mathcal{H}_{x} + \mathcal{H}_{y} - \mathcal{F}_{2} - \mathcal{F}_{3}- \mathcal{F}_{4} - \mathcal{F}_{6}, &\qquad 
					\mathcal{H}_{y} &= 2\mathcal{H}_{f} + \mathcal{H}_{g} - \mathcal{E}_{2} - \mathcal{E}_{4}  - \mathcal{E}_{6} - \mathcal{E}_{7}, \\ 
		\mathcal{E}_{1} &= \mathcal{F}_{7}, &\qquad 
							\mathcal{F}_{1} &= \mathcal{E}_{3},\\
		\mathcal{E}_{2} &= \mathcal{H}_{x} + \mathcal{H}_{y} - \mathcal{F}_{2} - \mathcal{F}_{3} - \mathcal{F}_{6}, &\qquad 
							\mathcal{F}_{2} &= \mathcal{H}_{f} + \mathcal{H}_{g} - \mathcal{E}_{2} - \mathcal{E}_{6} - \mathcal{E}_{7},\\
		\mathcal{E}_{3} &= \mathcal{F}_{1}, &\qquad 
							\mathcal{F}_{3} &=  \mathcal{H}_{f} + \mathcal{H}_{g} - \mathcal{E}_{2} - \mathcal{E}_{4} - \mathcal{E}_{6},\\
		\mathcal{E}_{4} &= \mathcal{H}_{x} - \mathcal{F}_{3}, &\qquad 
							\mathcal{F}_{4} &= \mathcal{H}_{f} - \mathcal{E}_{6},\\
		\mathcal{E}_{5} &= \mathcal{F}_{5}, &\qquad 
							\mathcal{F}_{5} &= \mathcal{E}_{5},\\
		\mathcal{E}_{6} &= \mathcal{H}_{x} + \mathcal{H}_{y} - \mathcal{F}_{2} - \mathcal{F}_{3} - \mathcal{F}_{4}, &\qquad 
							\mathcal{F}_{6} &= \mathcal{H}_{f} - \mathcal{E}_{2},\\
		\mathcal{E}_{7} &= \mathcal{H}_{x} - \mathcal{F}_{2}, &\qquad 
							\mathcal{F}_{7} &= \mathcal{E}_{1},\\
		\mathcal{E}_{8} &= \mathcal{F}_{8}, &\qquad 
							\mathcal{F}_{8} &= \mathcal{E}_{8}.
	\end{aligned}
\end{equation}

\begin{figure}[h]
\begin{equation}\label{eq:roots-q-Laguerre}
	\raisebox{-45pt}{\begin{tikzpicture}[
			elt/.style={circle,draw=black!100,thick, inner sep=0pt,minimum size=2mm}]
			\path (-0.6,-0.3) node (a1) [elt] {}
			( 0.6,-0.3) node (a2) [elt] {}
			( 0,0.67) node (a0) [elt] {}
			( -0.6,-1.2) node (a3) [elt] {}
			( 0.6,-1.2) node (a4) [elt] {};
            \draw [black] (a1) -- (a2) -- (a0) -- (a1);
			\draw [black,double distance=2pt] (a3) -- (a4);
            \node at ($(a1.west) + (-0.3,0.0)$) {$\alpha_{1}$};
            \node at ($(a2.east) + (0.3,0.0)$) {$\alpha_{2}$};
            \node at ($(a0.north) + (0,0.3)$) {$\alpha_{0}$};
            \node at ($(a3.south) + (0,-0.3)$) {$\alpha_{3}$};
            \node at ($(a4.south) + (0,-0.3)$) {$\alpha_{4}$};			
	\end{tikzpicture}}\quad 
			\begin{aligned}
			\alpha_{0} &= \mathcal{F}_{2} - \mathcal{F}_{1},	\\		
			\alpha_{1} &= \mathcal{H}_{y} - \mathcal{F}_{28}, \\
			\alpha_{2} &= 2 \mathcal{H}_{x} + \mathcal{H}_{y} - \mathcal{F}_{234567},  \\
			\alpha_{3} &= 2 \mathcal{H}_{x} + \mathcal{H}_{y} - \mathcal{F}_{123456},\\
			\alpha_{4} &=\mathcal{H}_{y} - \mathcal{F}_{67},\\
			\end{aligned}
	\qquad	\qquad 	
	\raisebox{-45pt}{\begin{tikzpicture}[
			elt/.style={circle,draw=black!100,thick, inner sep=0pt,minimum size=2mm}]
		\path 				(0,1)	node 	(d0) 	[elt, label=above:{$\delta_{0}$} ] {}
		        (-{sqrt(3)/2},1/2)	node 	(d1) 	[elt, label=left:{$\delta_{1}$} ] {}
		        (-{sqrt(3)/2},-1/2) node  	(d2)	[elt, label=left:{$\delta_{2}$} ] {}
		     	   			(0,-1)	node  	(d3) 	[elt, label=below:{$\delta_{3}$} ] {}
		        ({sqrt(3)/2},-1/2) 	node  	(d4) 	[elt, label=right:{$\delta_{4}$} ] {}
		        ({sqrt(3)/2},1/2) 	node 	(d5) 	[elt, label=right:{$\delta_{5}$} ] {};
		\draw [black,line width=1pt ] (d0) -- (d1) -- (d2) -- (d3) -- (d4) -- (d5)--(d0);
	\end{tikzpicture}}\quad 
			\begin{aligned}
			\delta_{0} &= \mathcal{F}_{6} - \mathcal{F}_{7},	\\		
			\delta_{1} &= \mathcal{H}_{x}+\mathcal{H}_{y} - \mathcal{F}_{1236} \\
			\delta_{2} &= \mathcal{F}_{3} - \mathcal{F}_{4},  \\
			\delta_{3} &=\mathcal{F}_{4} - \mathcal{F}_{5},\\
			\delta_{4} &= \mathcal{H}_{y}- \mathcal{F}_{34},\\
			\delta_{5} &= \mathcal{H}_{x}-\mathcal{F}_{68}.
			\end{aligned}
\end{equation}
	\caption{The final choice of the symmetry (left) and the surface (right) root bases for the $q$-Laguerre surface}
	\label{fig:roots-bases-qL}
\end{figure}

Applying the change of basis \eqref{eq:change-basis-final} to the root bases for the (b)-model surface on Figure~\ref{fig:roots-bases-b} we get the  
adjusted root bases for the $q$-Laguerre recurrence shown on Figure~\ref{fig:roots-bases-qL}.

Now the action of the mapping $\varphi$ on the symmetry roots $\alpha_{i}$ becomes
\begin{equation}\label{eq:q-Laguerre}
\varphi_*:\upalpha=\langle\alpha_0,\alpha_1,\alpha_2;\alpha_3,\alpha_{4}\rangle\mapsto
\varphi_*(\upalpha)=\upalpha+ \langle 0,1,-1;-1,1 \rangle\delta = (\phi_{*}\circ\psi_{*}) (\upalpha) = (\psi_{*}\circ\phi_{*}) (\upalpha).
\end{equation}

We can express $\varphi$ in terms of generators of the extended affine Weyl group $\widetilde{W}(E_{3}^{(1)})$ as
\begin{equation}\label{eq:qLag-decomp}
	\varphi = \sigma_{1}\circ \sigma_{0}\circ w_{3}\circ w_{0}\circ w_{2} = \psi \circ \phi = (w_{4}\circ \sigma_{2}^{3})\circ (\sigma_{2}^{4}\circ w_{0}\circ w_{2})
\end{equation}

\begin{remark} Given that $\varphi = \varphi_{2}^{-1}\circ \varphi_{1}$, see \eqref{eq:y-fwd}--\eqref{eq:x-back}, it is a natural question whether this decomposition 
	is the same as $\varphi =\psi \circ \phi$. A quick reflection suggest that this can not be the case, and indeed, we can compute that
	\begin{alignat*}{2}
		(\varphi_{1})_{*}(\upalpha) &= \langle-\alpha_0,\alpha_{0} + \alpha_2,\alpha_0 + \alpha_1;\alpha_4,\alpha_{3}\rangle, &\qquad  
		\varphi_{1} &= \sigma_{0}\circ \sigma_{2}^{2}\circ w_{0}\\
		(\varphi_{2})_{*}(\upalpha) &= \langle-\alpha_0,-\alpha_{1},2 \alpha_{0} + 2 \alpha_{1} + \alpha_2;\alpha_3 + 2\alpha_4,-\alpha_{4}\rangle, &\qquad 	
		\varphi_{2} &= \sigma_{0}\circ \sigma_{2}\circ w_{4} \circ w_{0}\circ w_{1}\circ w_{0},
	\end{alignat*}
	and so the backward map $\varphi_{2}$ acts on both components of the $E_{3}^{(1)}$ affine Dynkin diagram.
	
\end{remark}

The final step is to determine the explicit change of coordinates. For that, we also need to compute the root variables for the $q$-Laguerre case, 
which is done in the following Lemma.

\begin{lemma}\label{lem:q-Laguerre-root-vars} 
	The points $q_{1},\ldots,q_{8}$ lie on the polar divisor of the symplectic form $\omega$ that, in the affine $(x,y)$-chart, is given by 
	\begin{equation}\label{eq:q-Laguerre-sympl}
		\omega = k 	\frac{dx \wedge dy}{xy -1}
	\end{equation}
	\begin{enumerate}[(i)]
		\item The residues of the symplectic form $\omega$ along the irreducible components $d_{i}$ of the polar divisor corresponding to the 
		surface roots $\delta_{i} = [d_{i}]$, given in \eqref{eq:roots-q-Laguerre},  are 
		\begin{align*}
			\operatorname{res}_{d_{0}} \omega &= -k\, \frac{dv_{6}}{v_{6}-1}, &\qquad 
				 \operatorname{res}_{d_{2}} \omega &=k\, \frac{dv_{3}}{v_{3}(v_{3}-1)}, &\qquad   \operatorname{res}_{d_{4}} \omega &= k\,\frac{dx}{x},\\ 
			\operatorname{res}_{d_{1}} \omega &= k\, \frac{dy}{y}, &\qquad  \operatorname{res}_{d_{3}} \omega &= -k\, \frac{dv_{4}}{v_{4}},&\qquad   
			\operatorname{res}_{d_{5}} \omega &=  -k\, \frac{dy}{y},
		\end{align*}
		\item For the standard root variable normalization $\exp(\chi(\delta)) = a_{0}a_{1}a_{2} = a_{3}a_{4} =q$ we need to take 
		$k=-1$ and the root variables $a_{i}$ are then given by 
		\begin{equation}\label{eq:qLaguerre-root-vars}
			a_{0} = T,\quad a_{1} = \frac{1}{q^{n}T},\quad a_{2} = q^{n+1};\qquad a_{3} = q^{n+1}Q,\quad a_{4}=\frac{1}{q^{n}Q}.
		\end{equation}
		Note that the discrete time evolution $n\mapsto n+1$ induces the evolution 
		\begin{equation*}
			\overline{a}_{0} = a_{0},\quad \overline{a}_{1} = \frac{a_{1}}{q},\quad \overline{a}_{2} = q a_{2};\qquad \overline{a}_{3} = q a_{3},\quad \overline{a}_{4} = \frac{a_{4}}{q}
		\end{equation*}	
		matching the translation \eqref{eq:q-Laguerre}	(recall that the evolution of the root variables is dual to the evolution of the roots). 	
	\end{enumerate}
\end{lemma}

Proof is a direct computation and is omitted. We are now in the position to determine the birational change of variables identifying the 
$q$-Laguerre surfaces with the standard Sakai surface using the (b)-model. It is given in the following Lemma.

\begin{lemma}\label{lem:qLag-bPain-coord-change}
	The Sakai surface shown on Figure~\ref{fig:qLaguerre-pt-conf} can be matched with the (b)-model shown on Figure~\ref{fig:KNY-b-pt-conf} via
	the following explicit change of coordinates and parameter identification:
	\begin{equation}\label{eq:Charlier2Sakai-coord}
		\left\{
		\begin{aligned}
			x_{n}(f,g) &= - \frac{f_{n}}{g_{n} \nu_{2}\nu_{3}},\\
			y_{n}(f,g) &= - \frac{\nu_{2}\nu_{3}}{\kappa_{1}\kappa_{2}}\, \frac{\kappa_{1}\kappa_{2}g_{n}(f_{n}-\nu_{4})-\nu_{4}\nu_{6}\nu_{7}}{f_{n}^{2}},\\
			T&=\frac{\kappa_{1}\kappa_{2}}{\nu_{2}\nu_{3}\nu_{6}\nu_{7}},\quad  Q = \frac{\kappa_{1}\nu_{1}}{\nu_{2}\nu_{3}\nu_{8}},	\\
			q^{n} &= \frac{\nu_{2}\nu_{3}\nu_{4}\nu_{6}\nu_{7}\nu_{8}}{\kappa_{1}^{2}\kappa_{2}},		
		\end{aligned}
		\right. 
		\quad 
		\left\{
		\begin{aligned}
		f_{n}(x,y) &= \nu_{4}\frac{x_{n}-T}{T(x_{n}y_{n}-1)},\\
		g_{n}(x,y) &= -\frac{\nu_{4}}{\nu_{2}\nu_{3}}\frac{x_{n}-T}{T x_{n}(x_{n}y_{n} - 1)},\\
		\kappa_{1}&=\frac{\nu_{4}\nu_{8}}{q^{n}T},\quad  \kappa_{2} = \frac{\nu_{2}\nu_{3}\nu_{5}q^{2n+1}QT}{\nu_{4}},\\
		\nu_{1}&=\frac{q^{n}QT \nu_{2}\nu_{3}\nu_{5}}{\nu_{4}},\quad \nu_{6}=\frac{q^{n+1}Q\nu_{5}\nu_{8}}{T \nu_{7}},			
		\end{aligned}
		\right.
	\end{equation}
	As before, we can rescale the parameters $\nu_{2}\nu_{3}$ and $\nu_{4}$ to be equal to $1$, the remaining parameters $\nu_{i}$ 
	do not play a role in the variable change and get canceled in the birational representation.
\end{lemma}

\begin{proof}
	The parameter matching is done via the root variables expressions \eqref{eq:root-vars-KNY-b} and \eqref{eq:qLaguerre-root-vars}. 
	We get the explicit change of coordinates from the change of basis \eqref{eq:change-basis-final} in the Picard lattice. Let us briefly recall the key steps
	using the simplest case.
	To find $x(f,g)$ we need to find the basis of the pencil of $(1,1)$-curves $|H_{x}| = |H_{f} + H_{g} - \mathcal{E}_{2} - \mathcal{E}_{6}|$ on the $(f,g)$-plane
	passing through the points $p_{2}(\infty,\infty)$ and $p_{6}(0,0)$ on Figure~\ref{fig:KNY-b-pt-conf}, which has the form $a_{10}f + a_{01} g = 0$. 
	A homogeneous coordinate on this pencil is $[f:g]$, which coincides with $x$ up to some M\"obius transformation. Thus, 
	\begin{equation*}
		x(f,g) = \frac{A f + B g}{C f + Dg}.
	\end{equation*}
	We determine the coefficients $A,\ldots,D$ using the mapping of the exceptional divisors. For example, 
	\begin{equation*}
		[(f=0)] = \mathcal{H}_{f} - \mathcal{E}_{5} - \mathcal{E}_{6} = \mathcal{F}_{4} - \mathcal{F}_{5},
	\end{equation*}
	so in affine charts the map should collapse $f=0$ to $q_{3}(0,\infty)\leftarrow q_{4}$. Thus, $B=0$. Similarly, $g=0$ should map to $x=\infty$, and so $C=0$. 
	Thus, $x(f,g) = (Af)/(Dg) = A (f/g)$, since we can, without loss of generality, put $D=1$. To find $A$ we note that $p_{3}$ should map to $q_{1}$, 
	and thus $x(f,g)(p_{3}) = A (f/g)(p_3) = A v_{3}(p_3) = A(-\nu_{2}\nu_{3}) = 1$, so $A=-(\nu_{2}\nu_{3})^{-1}$. Other computations are done along the same lines, 
	but are more involved, see \cite{DzhFilSto:2020:RCDOPWHWDPE} or \cite{HuDzhChe:2020:PLUEDPE} for more detailed examples. 	
\end{proof}

\begin{remark} Using the birational representation of $\widetilde{W}(E_{3}^{(1)})$, we can write the mapping 
	$\varphi = \sigma_{1}\circ \sigma_{0}\circ w_{3}\circ w_{0}\circ w_{2}$ explicitly. For simplicity, we write it using the root variables
	of the (b)-model, but it is still quite complicated, 
	\begin{equation}\label{eq:qLag-std}
		\begin{aligned}
		\varphi(f,g) &= \left(
			\frac{f(a_{0}a_{2}(f(1-g) + a_{1}(a_{2}g - 1))+ a_{4} g(1 - a_{2} g))}{a_{0}a_{1}a_{2}g(a_{2}g-1)(a_{0}a_{2}f - a_{4}(1 - a_{2}g))}, \right.\\
			&\quad\left. \frac{f(a_{0}a_{2}f(g-1)+a_{4}g(a_{2}g-1))(a_{0}a_{2}(f(g-1)-a_{1}(a_{2}g-1))+a_{4}g(a_{2}g-1))}{
			a_{4}(a_{2}g-1)(a_{0}a_{2}f(f(g-1)+a_{1}(a_{2}g-1)(a_{0}a_{2}g-1))+a_{4}g(a_{2}g-1)(f + a_{0}a_{1}a_{2}(a_{2}g-1)))}
		\right)
		%
		%
		\end{aligned}
	\end{equation}
	It is straightforward to verify that the change of coordinates 	\eqref{eq:Charlier2Sakai-coord} transforms this mapping into our 
	recurrence \eqref{eq:qxn-back-for}, when expressed in the evolutionary form. The dynamic in parameters $\kappa_{i}$ and $\nu_{j}$ can be
	obtained using \eqref{eq:root-vars-KNY-b}.	
\end{remark}	
	
\section{Conclusions} 
\label{sec:conclusions}
In this paper we studied a recurrence relation appearing in the study of deformed $q$-Laguerre orthogonal polynomials and identified it as a discrete Painlev\'e 
equation on the $A_{5}^{(1)}$ family of Sakai surfaces. An interesting feature of this example is that this equation in a composition of two standard mappings, and
so recognizing it as such without the use of geometric techniques may be very difficult, if not impossible. In contrast, the geometric identification procedure
makes this process quite straightforward and effectively removes any guesswork. Given that this deformed $q$-Laguerre dynamic is a composition of a 
discrete $q$-$\Pain{IV}$ and discrete $q$-$\Pain{III}$ equations, it is an interesting follow-up question to see if it has a good continuous limit. 

\appendix

\providecommand{\bysame}{\leavevmode\hbox to3em{\hrulefill}\thinspace}
\providecommand{\MR}{\relax\ifhmode\unskip\space\fi MR }
\providecommand{\MRhref}[2]{%
  \href{http://www.ams.org/mathscinet-getitem?mr=#1}{#2}
}
\providecommand{\href}[2]{#2}

\end{document}